\def\showcomments{0}
\documentclass[acmsmall]{acmart}

\setcopyright{rightsretained}
\acmDOI{10.1145/3704870}
\acmYear{2025}
\acmJournal{PACMPL}
\acmVolume{9}
\acmNumber{POPL}
\acmArticle{34}
\acmMonth{1}
\received{2024-07-11}
\received[accepted]{2024-11-07}

\usepackage{stmaryrd} 
\usepackage{enumerate}
\usepackage[inline]{enumitem}
\usepackage{xspace}
\usepackage{cuted}
\usepackage{tabularx}
\usepackage{tikz}
\usetikzlibrary{shapes.misc,arrows.meta}
\usepackage{tikz-qtree}
\usepackage{thmtools}
\usepackage{bold-extra}
\usepackage{subcaption}
\usepackage[frozencache=true,cachedir=minted-cache]{minted}
\usepackage{float}

\theoremstyle{definition}
\newtheorem{theorem}{Theorem}[section]
\newtheorem{corollary}[theorem]{Corollary}
\newtheorem{lemma}[theorem]{Lemma}
\newtheorem{proposition}[theorem]{Proposition}
\newtheorem{definition}[theorem]{Definition}

\newtheorem{example}[theorem]{Example}

\newenvironment{myparagraph}[1]{\vspace{0.1in}\noindent{\bf #1.}}{}

\newcommand{\figlabel}[1]{\label{fig:#1}}
\newcommand{\figref}[1]{Fig.~\ref{fig:#1}}
\newcommand{\seclabel}[1]{\label{sec:#1}}
\newcommand{\secref}[1]{Section~\ref{sec:#1}}

\newcommand{\deflabel}[1]{\label{def:#1}}
\newcommand{\defref}[1]{Definition~\ref{def:#1}}
\newcommand{\thmlabel}[1]{\label{thm:#1}}
\newcommand{\thmref}[1]{Theorem~\ref{thm:#1}}

\newcommand{\propref}[1]{Proposition~\ref{prop:#1}}
\newcommand{\lemlabel}[1]{\label{lem:#1}}
\newcommand{\lemref}[1]{Lemma~\ref{lem:#1}}
\newcommand{\corlabel}[1]{\label{cor:#1}}
\newcommand{\corref}[1]{Corollary~\ref{cor:#1}}

\makeatletter
\renewenvironment{proof}[1][\proofname]{\par
  \pushQED{\qed}%
  \normalfont\topsep6\p@\@plus6\p@\relax
  \trivlist
  \item[\hskip\labelsep\itshape
    #1\@addpunct{.}]\ignorespaces
}{%
  \popQED\endtrivlist\@endpefalse
}
\makeatother

\ifnum\showcomments=1
\newcommand{\um}[1]{{\sf[{\footnotesize{UM: {\color{red} #1}}]}}}
\newcommand{\mv}[1]{{\sf[{\footnotesize{MV: {\color{cyan} #1}}]}}}
\newcommand{\dm}[1]{{\sf[{\footnotesize{DM: {\color{blue} #1}}]}}}
\else
\newcommand{\um}[1]{}
\newcommand{\mv}[1]{}
\newcommand{\dm}[1]{}
\fi

\newcommand{\revise}[1]{#1}

\newcommand{\mlq} {\text{`}}
\newcommand{\mrq} {\text{'}}

\newcommand{\nats}{\mathbb{N}}
\newcommand{\set}[1]{\{#1\}}
\newcommand{\setpred}[2]{\{#1 \,\mid\, #2\}}
\newcommand{\sem}[1]{\llbracket #1 \rrbracket}

\renewcommand{\emptyset}{\varnothing}
\newcommand{\xdownarrow}[1]{%
  {\left\downarrow\vbox to #1{}\right.\kern-\nulldelimiterspace}
}

\newcommand{\powset}[1]{\mathcal{P}(#1)}

\newcommand{\vocab}{\Pi}
\newcommand{\consts}{\mathcal{C}}
\newcommand{\funcs}{\mathcal{F}}
\newcommand{\rels}{\mathcal{R}}
\newcommand{\vars}{\mathcal{V}}
\newcommand{\conns}{\mathcal{O}}
\newcommand{\arity}[1]{\mathsf{arity}(#1)}
\newcommand{\model}{\mathcal{M}}
\newcommand{\interp}[1]{\sem{#1}}
\renewcommand{\terms}[1]{\mathsf{Terms}_{#1}}
\newcommand{\tru}{\mathtt{true}}
\newcommand{\fls}{\mathtt{false}}
\newcommand{\asgn}{\alpha}

\newcommand{\Ss}{\mathcal{S}}
\newcommand{\Aa}{\mathcal{A}}

\newcommand{\MM}{\mathcal{M}}
\newcommand{\qacc}{q_\text{acc}}

\renewcommand{\phi}{\varphi}

\newcommand{\allc}{\text{all}}
\newcommand{\nonec}{\text{none}}
\newcommand{\epr}{[\exists^*\forall^*,\allc,(0)]_=}
\newcommand{\eprweak}{[\exists^*\wedge \forall^3,(0,2),(0)]_=}
\newcommand{\gurevich}{[\exists^*,\allc,\allc]_=}
\newcommand{\gurevichweak}{[\exists^*,(0),(2)]_=}

\newcommand{\alphabet}{\Sigma}
\newcommand{\nodes}{\text{nodes}}

\newcommand{\shallowterms}{\textsf{ShallowTerms}}
\newcommand{\st}{\shallowterms}
\newcommand{\congr}{\textsf{Cong}}
\newcommand{\val}{\textsf{val}}

\newcommand{\assume}{\textbf{assume}}
\newcommand{\comp}{\text{Comp}}

\newcommand{\core} {\textsf{coRE}}
\newcommand{\re} {\textsf{RE}}
\newcommand{\nexp}{\textsf{NEXP}}

\newcommand{\kw}[1]{\texttt{\textbf{#1}}}
\newcommand{\cd}[1]{\texttt{#1}}
\newcommand{\tuple}[1]{\langle{#1}\rangle}
\newcommand{\ntrm}[1]{\tuple{\texttt{#1}}}
\newcommand{\expr}{\ntrm{e}}
\newcommand{\cond}{\ntrm{b}}
\newcommand{\stmt}{\ntrm{stmt}}
\newcommand{\skipk}{\kw{skip}}
\newcommand{\ifk}{\kw{if}}
\newcommand{\thenk}{\kw{then}}
\newcommand{\elsek}{\kw{else}}
\newcommand{\whilek}{\kw{while}}
\newcommand{\assertk}{\kw{assert}}

\newcommand{\grmsep}{\,\,\,|\,\,\,}

\newcommand{\progstmt}[1]{%
    \tikz[baseline={(word.base)}]{%
        \node[draw=blue!50!black,rectangle,rounded corners=3pt,thick, fill=gray!10, inner sep=0.5mm] (word) {#1};}%
}

\begin{document}

\title{The Decision Problem for Regular First Order Theories}

\author{Umang Mathur}
\orcid{0000-0002-7610-0660}
\affiliation{%
  \institution{National University of Singapore}
  \city{Singapore}
  \country{Singapore}
}
\email{umathur@comp.nus.edu.sg}

\author{David Mestel}
\orcid{0000-0003-3186-8307}
\affiliation{%
  \institution{Maastricht University}
  \city{Maastricht}
  \country{Netherlands}
}
\email{david.mestel@maastrichtuniversity.nl}

\author{Mahesh Viswanathan}
\orcid{0000-0001-7977-0080}
\affiliation{%
  \institution{University of Illinois at Urbana-Champaign}
  \city{Urbana}
  \country{USA}
}
\email{vmahesh@illinois.edu}


\begin{abstract}
The \emph{Entscheidungsproblem}, or the classical decision problem, asks whether a given formula of first-order logic is satisfiable.  In this work, we consider an extension of this problem to regular first-order \emph{theories}, i.e., (infinite) regular sets of formulae.  Building on the elegant classification of syntactic classes as decidable or undecidable for the classical decision problem, we show that some classes (specifically, the EPR and Gurevich classes), which are decidable in the classical setting, become undecidable for regular theories. On the other hand, for each of these classes, we identify a subclass that remains decidable in our setting, leaving a complete classification as a challenge for future work.  Finally, we observe that our problem generalises prior work on automata-theoretic verification of uninterpreted programs and propose a semantic class of existential formulae for which the problem is decidable.
\end{abstract}



\maketitle


\section{Introduction}

The classical `decision problem' can be stated in one of two ways that are duals of one another: the \emph{satisfiability problem} which asks if a given first order logic formula is satisfiable, or the \emph{validity problem} which asks if a given formula is valid. Since the validity problem is known to be {\re}-complete~\cite{god30,chu36,tur37} (satisfiability is {\core}-complete), recent research has focused on identifying fragments of first order logic for which the classical decision problem is decidable, and when decidable, determining the corresponding complexity. This remains a central challenge in computer science and logic, having inspired a wealth of results and techniques that have found applications in solving other problems. Most questions in this area have been resolved, with an excellent summary provided in~\cite{classical-book}.

In this paper we initiate the study of the classical decision problem for theories of first order formulae. 
In other words, the problem we investigate is, given a set of formulae (effectively presented), 
determine if the set is satisfiable, or determine if the set is valid --- a theory is \emph{satisfiable} if there is a model such that every formula in the theory is true in the model, while a theory is \emph{valid} if all formulae in the theory are true in every model.  (For theories these are distinct problems, in contrast to the situation for single formulae where $\phi$ is valid if and only if $\neg \phi$ is unsatisfiable.) The validity problem for theories can be equivalently expressed as asking (after negating each formula) if there is a model and a formula in the theory such that the formula is true in the model, and so can be viewed as a different form of satisfiability. In this paper we study both these forms: 
\emph{conjunctive satisfiability} asks if there is model $\model$ for which all formulae from a given set $T$ of first order formulae are true (and is denoted $\model \models \bigwedge T$), and 
\emph{disjunctive satisfiability} asks if there is a model in which some formula from a given set of formulae is true (and is denoted $\model \models \bigvee T$). 

The conjunctive and disjunctive satisfiability problems for first order theories arise naturally in a number of contexts. Classically, the entailment problem, which asks if a formula $\varphi$ is entailed by a theory $\Gamma$ (i.e., `does $\Gamma \models \varphi$?'), is equivalent to asking if the set $\Gamma' = \setpred{\neg\varphi \wedge \psi}{\psi \in \Gamma}$ is conjunctively satisfiable. In verification (or bug finding), one is interested in determining if some execution of a given program violates the requirements. Associating with each execution $\rho$ a formula $\varphi_\rho$ that expresses the condition that $\rho$ violates the specification, the problem of bug detection can be reduced to the problem of checking the disjunctive satisfiability of the set consisting of all the formulae associated with executions of the given program~\cite{Mathur2019}. In the context of compiler transformations, `executions' are programs themselves, and correctness of compiler transformations can also be reduced to disjunctive satisfiability~\cite{equalitysaturation2009}. 
In synthesis, the goal is often to identify a syntactic object (program, invariant, formula) from a set of possible programs/invariants/formulae that meets certain requirements. 
Once again, by associating with each possible syntactic entity a formula that characterizes the conditions under which it meets the requirements, synthesis can be reduced to disjunctive satisfiability as well~\cite{srivastava2013template,coherent-synthesis2020}. 
In \secref{motivating}, we discuss, in detail, the connection
between these and other
problems arising in logic and verification with the problem we consider here. 
Finally, in software testing, the goal is often to design a test suite that is diverse and 
exercises different paths of the program under test. 
For example, an effective test suite for an SMT solver must 
have formulae that are both satisfiable and unsatisfiable~\cite{SMT-testing-2020}.
Likewise, an effective test suite for testing database management systems
must contain SQL queries that 
return different types of answers~\cite{Rigger2020}.
Similarly, a test suite for a concurrent program must exercise different interleavings~\cite{fuzzingCCT2024}. 
By associating with every test, a formula  that characterizes its type
(satisfiable/unsatisfiable, part of code exercised), 
we can reduce the problem of checking if a given test suite is diverse
to the problem of checking disjunctive satisfiability of a theory.

There are different ways in which an infinite set of formulae could be presented effectively.
In this paper we consider regular first order theories. 
In other words, a theory $\Gamma$ is represented by a tree automaton that 
accepts the parse tree of a formula if and only if the formula belongs to $\Gamma$. 
For such theories, we ask when conjunctive and disjunctive satisfiability
of theories become decidable. 
Clearly, fragments of first order logic for which the classical decision problem (for single formulae) \cite{classical-book} is known to be undecidable,  will continue to be undecidable in this more general case as well. Hence, we primarily focus on decidable fragments of first order logic
and present the first results, leaving a more complete characterization to future work.


\subsection{Our Contributions}
\seclabel{contrib}

Our first contribution is to adapt the main tool for proving decidability for the classical 
decision problem, the so-called `finite model property', to the setting of regular theories.
We observe that the finite model property is not sufficient for decidability for regular 
theories, but we define strengthenings of it, which we call the weak and strong 
\emph{bounded model property} (see~\defref{bounded-model-property}), 
and show that these are sufficient to establish decidability.



\begin{restatable}{theorem}{boundedModelProperty}
\thmlabel{boundedmodel}
Let $\mathcal{C}$ be a class of formulae.
If $\mathcal{C}$ has the weak bounded model property, then the conjunctive satisfiability problem for regular 
$\mathcal{C}$-theories is decidable.  
If $\mathcal{C}$ has the strong bounded model property,
then the disjunctive satisfiability problem for regular 
$\mathcal{C}$-theories is also decidable.
\end{restatable}

Crucial to the proof of~\thmref{boundedmodel} 
is the decidability of the model checking problem (\corref{modelcheck}),
which asks to determine if 
$\model \models \bigwedge T$ or $\model \models \bigvee T$
for a given finite model $\model$.

We next consider two syntactic classes for which the classical decision problem is decidable: 
the Bernays-Sch{\"o}nfinkel
class of function-free formulae with quantifiers of the form $\exists \ldots \exists 
\forall \ldots \forall$ (also popularly known as the \emph{effectively propositional}, or EPR class), 
and the Gurevich class of purely existential formulae.  
We show that both conjunctive and disjunctive satisfiability problems
are undecidable for both of these classes.
To establish this, we show how each class here contains an undecidable subclass.
In the following, we use notations such as $[\exists^* \forall^3,(0,2),(0)]_=$
to denote fragments of first order logic, and are explained in detail in \secref{fologic}.

\begin{restatable}{theorem}{eprundec}
\thmlabel{eprundec}
    The disjunctive and conjunctive satisfiability problems for regular theories are undecidable for the class 
    $\eprweak$, and hence for the standard syntactic class $[\exists^* \forall^3,(0,2),(0)]_=$ (which is contained in the EPR class).
\end{restatable}

\begin{restatable}{theorem}{exundec}
    \thmlabel{exundec}
    The disjunctive and conjunctive satisfiability problems for regular theories are undecidable for the class $[\exists^*, (0),(2)]_=$.
\end{restatable}

We first show that some fragments of the above classes nevertheless admit decidable 
satisfiability of theories.  
Firstly, the classes of purely existential and purely universal function-free formulae 
admit decidable for both the conjunctive and the disjunctive satisfiability problems for theories.
We remark that the former is a subclass of both the EPR and Gurevich class, 
while the latter is a subclass of the EPR class.

\begin{restatable}{theorem}{funcfree}
    \thmlabel{funcfree}
    The disjunctive and conjunctive satisfiability problems for regular theories are decidable for the classes $[\exists^*,\allc,(0)]_=$ 
    and $[\forall^*,\allc,(0)]_=$.
\end{restatable}

Secondly, the class of quantifier-free formulae with functions (a subclass of the Gurevich 
class) admits decidable conjunctive satisfiability.  We believe that disjunctive satisfiability
is also decidable for this class, but leave this open for future work.

\begin{restatable}{theorem}{quantfreecon}
\thmlabel{quantfreecon}
The conjunctive satisfiability problem for regular theories is decidable for the class $[\nonec,\allc,\allc]_=$ of quantifier-free formulae with equality.
\end{restatable}

Finally, beyond the syntactic classes outlined above, 
we also identify a semantic class of existentially quantified formulae 
with functions for which the disjunctive satisfiability problem is decidable.
This class of formulae, which we call \emph{coherent formulae},
is inspired by recent work that identifies a subclass of uninterpreted programs
called \emph{coherent programs}~\cite{Mathur2019,Mathur2020axioms,Mathur2020heapmanipulating,coherent-synthesis2020} 
for which program verification becomes decidable.
Programs in this class are such that all their program paths (or executions)
retain sufficient information `in-scope'.
In turn, this ensures that the necessary task of computing the congruence closure, 
induced by equality assumptions observed in each execution, 
can be computed on-the-fly using a constant space streaming algorithm.
Such a streaming algorithm can then be translated to an automata-theoretic
decision procedure for verification of programs in this class (coherent programs).
In this work, we generalize the essence of this streaming algorithm,
and extend it to the case of arbitrary existential formulae 
(represented as trees).
This allows us to identify the semantic class of \emph{coherent formulae}
whose satisfiability can be checked soundly using a
finite-state decision procedure.

\begin{restatable}{theorem}{coherentDec}
\thmlabel{coherentDec}
    Disjunctive satisfiability is decidable for regular coherent existential theories.
\end{restatable}

\noindent
{\bf Remark}. This is an extended version of the paper accepted to appear in the proceedings of POPL 2025~\cite{FOTheories2025}.


\section{Regular Decision Problem in Practice}
\seclabel{motivating}

Although the investigation in this work is primarily theoretical,
we discuss and illustrate natural connections between the problem 
of satisfiability checking for regular theories
and several applications in program analysis, verification and synthesis.


\begin{figure}[t]
\begin{subfigure}[b]{0.4\textwidth}
\scriptsize
\begin{align*}
\begin{array}{rcl}
\expr\!\!\!\!\!&:=&\!\!\!\!\!x \grmsep c \grmsep f(\expr, \ldots, \expr)
\\
\cond\!\!\!\!\!&:=&\!\!\!\!\!\expr=\expr \grmsep R(\expr, \ldots, \expr) \grmsep \neg \cond \grmsep \cond \lor \cond \\
\stmt\!\!\!\!\!&:=&\!\!\!\!\!\skipk \grmsep x \gets \expr \grmsep \ifk \, \cond \, \thenk \, \stmt \, \elsek \, \stmt \\
&& \grmsep \assertk(\cond) \grmsep \whilek \, \cond \, \stmt \grmsep \stmt;\stmt \\
\\
&& x \in \mathcal{X}, c \in \mathcal{C}, f \in \mathcal{F}, R \in \mathcal{R}
\end{array}
\end{align*}
\caption{Grammar for imperative programs}
\figlabel{grammar}
\end{subfigure}
\hfill
\begin{subfigure}[b]{0.45\textwidth}
\begin{center}
\begin{minted}[breaklines=true,xleftmargin=10pt,fontsize=\small]{C}
x:=0;
while (x < 10){
	x := x+1
}
@post: assert(x==10)
\end{minted}
\end{center}
\caption{Example program $P$ over the grammar}
\figlabel{prog}
\end{subfigure}
\caption{Imperative programs}
\figlabel{imp}
\end{figure}

\myparagraph{Algorithmic verification of imperative programs}{
	Consider the verification problem for imperative programs
	with loops, operating over a potentially unbounded data domain.
	\figref{grammar} shows a simple grammar consisting of usual constructs
	like \ifk-\thenk-\elsek, \whilek\ and sequencing,
	and whose expressions comprise of a finite set of variables $\mathcal{X}$, constants $\mathcal{C}$,
	function symbols $f \in \mathcal{F}$ and relation symbols $\mathcal{R}$.
	\figref{prog} shows an example program over this grammar.
	While undecidable in general, a popular paradigm in algorithmic verification
	resorts to a language-theoretic view of this problem,
	and has notably led
	to decidability results~\cite{Mathur2019,AlurCerny2011,RegularModelChecking2000,Mathur2020heapmanipulating,Mathur2020axioms}, as well as state-of-the-art verification tools~\cite{traceAbstraction2009} and semi-decision procedures~\cite{Farzan2020}.
	Here, one models a program as a set (or language) of words
	representing program paths, 
	which is often a regular set over some  fixed vocabulary.
	Consider, for example, the program $P$ in \figref{prog}.
	The set of the program paths of $P$
	can be characterized by the following regular language:
	\[\mathsf{Paths}_P = \progstmt{\cd{x:=0}} \cdot \Big(\progstmt{\kw{assume}\cd{(x < 10)}}\cdot \progstmt{\cd{x := x+1}}\Big)^*.\]
	Here the alphabet 
	$\Pi = \set{\progstmt{\cd{x:=0}}, \progstmt{\kw{assume}\cd{(x < 10)}}, \progstmt{\cd{x := x+1}}}$ consists of individual statements, which are either
	assignments (derived directly from the program syntax) or assume statements
	(compiled down from conditionals and loop guards).
	Besides, such a regular language description of the set of program paths
	can be effectively obtained.
	Next, one asks if the subset of the \emph{feasible} runs
	in the set intersects with the language $L_{\sf bug}$ of runs 
	witnessing a buggy behavior.
	Often $L_{\sf bug}$ is also a regular set.
	However, the set of feasible runs is, in general, an undecidable set.
	When a class $\mathcal{P}$ of programs is such that
	for every program in the class $\mathcal{P}$, 
	either the set of feasible runs of the program is
	regular, or more generally, 
	when the emptiness of this set, intersected with a regular set (namely $L_{\sf bug}$),
	can be checked in a decidable manner, the verification problem for 
	$\mathcal{P}$ also becomes decidable~\cite{Mathur2019,AlurCerny2011,Govind2021}.
	
	An alternative, but still language-theoretic, 
	formulation of the verification problem is through sets of \emph{formulae}.
	For each program path (or run) $\sigma$ in a program $P$,
	one can construct a first-order logic 
	formula $\varphi_{\sigma}$ that symbolically encodes 
	the question --- 
	`is there an initial valuation of the variables that makes $\sigma$ a feasible path'?
	Representing each formula as a tree, the set of formulae $\mathcal{T}_P$ for $P$ 
	can be shown to be tree-regular, when $P$ consists of simple imperative
	constructs as in~\figref{grammar}.
	For the program $P$ in \figref{prog}, the set 
	$\mathcal{T}_P$ of formula trees of $P$ is given by the parse trees
	of the following context-free grammar:
	\begin{align*}
	\begin{array}{rcl}
	S &\rightarrow& \cd{x} = 0 \land T \\
	T &\rightarrow& \cd{x} = 10 \grmsep \exists \cd{x'}\, .\, \cd{x'} = \cd{x} \land (\exists \cd{x}\, .\, \cd{x} = \cd{x'}+1 \land T) 
	\end{array}
	\end{align*}
	The set of formulae described through the grammar above
	essentially captures, for each run $\sigma$ of $P$,
	the symbolic constraint that encodes the feasibility of $\sigma$.
	This means, the program verification question 
	translates to the following question about a set of formulae --- 
	\emph{is there a formula $\varphi$ in $\mathcal{T}_P$ such 
	that $\varphi \land \psi_{\text{bug}}$ is satisfiable?}
	Here $\psi_{\text{bug}}$ is a formula that characterizes a buggy behavior.
	Clearly, an algorithm that solves the regular decision problem
	can be used to solve the program verification question, and hence decidability results for the former problem can be 
	translated to the latter.
	This is a central motivation of the investigation we carry out in this work.
}


\myparagraph{Program synthesis and unrealizability}{
	The syntax-guided synthesis (SyGuS) problem~\cite{sygus2013} 
	asks to generate a program from a given 
	grammar $\mathcal{G}$ (typically generating a class of loop-free programs)
	that satisfies a given specification $\varphi$,
	often specified as a list of pairs of input and output examples.
	The set of terms (or programs) accepted by $\mathcal{G}$
	is tree-regular and further the specification can be conjuncted with formulae 
	characterizing the semantics of these programs.
	This gives a natural translation of the SyGuS to the 
	regular decision problem --- the SyGuS problem admits a solution iff the tree regular set of formulae thus constructed contains a satisfiable formula.
	Notably, when the set of tree accepted by $\mathcal{G}$ is infinite,
	the realizability problem `is there a valid program?' becomes
	undecidable in general~\cite{coherent-synthesis2020,caulfield2016whatsdecidablesyntaxguidedsynthesis}, 
	but nevertheless decidable fragments of the regular decision problem can yield analogous decidability results
	for the (un)realizability problem as well~\cite{Hu-unrealizability-2019,hu-unrealizability-2020}.
}

\myparagraph{Logic Learning}{
	An emerging trend in the area of safe machine learning
	is to study algorithms for approximating the behavior of a classifier
	with a more systematic representation, and in particular, 
	a logical formula~\cite{Krogmeier2022,Krogmeier2023,Koenig2020}.
	Here, one is given a finite set
	of first-order structures (over a fixed vocabulary) marked
	positively or negatively, and
	also a logic $\mathcal{L}$ (either full FOL, or a sub-class defined syntactically using a grammar),
	and the task is to determine
	if there is a formula $\varphi \in \mathcal{L}$
	that evaluates to true over the positive structures, and false on the negative structures.
	When $\mathcal{L}$  is presented as a grammar, then the set of formulae
	that evaluate to true on a given positive example becomes a regular set,
	and a solution to the regular decision problem for this set also solves
	the realizability of the learning problem~\cite{Krogmeier2023}.
	Further, when the algorithm for solving the decision problem
	is also automata-theoretic, then the synthesis problem 
	(i.e., output a formula, if one exists) can also be solved effectively,
	and for multiple positive and negative examples at once.
}


\section{Preliminaries}

\subsection{First order logic and its fragments}\label{sec:fologic}

A first order (FO) vocabulary is a triple $\vocab = (\consts, \funcs, \rels)$,
where $\consts$ is a set of constant symbols, and
$\funcs = \uplus \funcs_i$ and $\rels = \uplus \rels_i$ 
are fixed sets of 
function and relation symbols, indexed by their arities.
$\vocab$ is finite if each of $\consts, \funcs, \rels$ is finite.
For a function symbol $f \in \funcs$ and a relation symbol $R \in \rels$, 
we denote by $\arity{f}$ and $\arity{R}$ the arities of $f$ and $R$ respectively.
The set of terms
that can be formed using a set of variables $\vars$ and symbols from $\vocab$, 
denoted $\terms{\vocab, \vars}$, is given by the following BNF grammar 
(here $x \in \vars, c \in \consts$ and $f \in \funcs$ is a function symbol of appropriate arity):
\begin{align*}
\begin{array}{rcl}
t &::=& c \; | \; x \; | \; f(t, \ldots, t)
\end{array}
\end{align*}
The set of formulae is constructed using equality atoms (`$t_1 = t_2$'),
relational atoms (`$R(t_1, \ldots, t_k)$'), or, 
inductively using existential/universal quantification
or using propositional connectives such as `$\neg$', `$\land$' and `$\lor$'.
The BNF grammar for the set of FO formulae is:
\begin{align*}
\begin{array}{rcl}
\varphi &::=& t = t \; | \; R(t, \ldots, t) \; |\; \neg \varphi \; |\; \varphi \land \varphi \; |\; \varphi \lor \varphi \; |\;  \exists x \cdot \varphi \; |\; \forall x \cdot \varphi \\
\end{array}
\end{align*}
Here, $R \in \rels$ is a relation symbol with appropriate arity. 
We will sometimes refer to the vocabulary 
of a formula $\varphi$, together with the finite
set of variables in it, as its \emph{signature}, $\text{sig}(\varphi)$.


The semantics of first order logic over vocabulary $\vocab$
is given using a first order structure (or model), which is a tuple
$\model = (U, \interp{})$, where $U$ is the universe and
$\interp{}$ is the interpretation that maps every constant symbol in $c$
to an element $\interp{c} \in U$, every function symbol $f$ of arity $r$
to a function $\interp{f} \in [U^r \to U]$
and every relation symbol of 
arity\footnote{We also allow $0$-ary relations. 
In this case, the interpretation $\interp{}$ maps such a relation to either $\tru$ or $\fls$} 
$r$ to a subset of $\interp{R} \subseteq U^r$.
An assignment $\asgn : \vars \to U$ maps variables to elements of the universe.
The evaluation of a first order logic formula $\varphi$ 
over model $\model$ and an assignment $\asgn$ 
can be described inductively by means of 
a satisfaction relation ($\model, \asgn \models \varphi$), is standard, and is skipped. 
An FO formula $\varphi$ is said to be satisfiable if there is an FO structure
$\model$ and an assignment $\asgn$ such that $\model, \asgn \models \varphi$, and valid
if $\model,\asgn \models \varphi$ holds for every FO structure $\model$ and assignment $\asgn$.
When $\varphi$ is a sentence (i.e., it contains no free variables),
then we simply write $\model \models \varphi$.
The size of the model $\model = (U, \interp{})$ is the cardinality $|U|$.
A class of models is said to be bounded by $k \in \nats$, if
each of their sizes is bounded by $k$.

\myparagraph{Syntactic fragments of first order logic}{
The decision problem for first-order logic is undecidable.
Thus, in order 
to identify decidable cases, it is necessary to restrict attention to fragments of first order logic.  
In this work, we will mainly consider `prefix-vocabulary classes', which 
classify formulae by (1) the number and arity of function symbols in the signature, (2) the number and arity 
of relation symbols, (3) the presence or absence of equality and (4) the permitted order of quantifiers.

Briefly (see \cite{classical-book} Definition 1.3.1 for full details), such a class is written as $[\Pi,p,f]$ 
(or $[\Pi,p,f]_=$ if equality is present), where $p$ and $f$ are 
`arity sequences' from $(\nats^+\cup \{\omega\})^\nats$ whose $k^{\text{th}}$ entry denotes the available number of predicates 
(respectively functions) of arity $k$, and $\Pi$ is a set of quantifier prefixes, i.e. strings over $\{\forall, \exists\}$ 
denoting the allowed order of quantifiers.
We say that $p$ or $f$ is `standard' if whenever infinitely many symbols of arity $k$ or above are available, 
the $k^{\text{th}}$ entry of $p$ (respectively $f$) is $\omega$ (which makes sense since we can always use a 
higher arity symbol as a lower arity one).  Note that if $p$ (respectively $f$) is standard then unless it is 
equal to $\omega^\nats$, which we denote `$\allc$', then it contains only finitely many non-zero entries, and 
we omit the trailing zeroes.
We say that $\Pi$ is standard if it is either the set of all prefixes 
(denoted $\allc$) or a set which can be expressed as a finite string in the letters $\{\forall, \exists, \forall^*, \exists^*\}$ 
(and if pairs like $\forall^* \forall$ have been absorbed as $\forall^*$).  The class $[\Pi,p,f]$ or $[\Pi,p,f]_=$ is standard 
if $\Pi$, $p$ and $f$ are standard.\footnote{When considering single formulae, it is usual 
to assume that all variables are bound and that there are no constants, as these are equivalent
to existentially bound variables.  This is not the case for theories and so we will not make 
this assumption.}

Since the set of standard classes with containment is a well-quasi-ordering~\cite{gurevich1969wqo}, there is guaranteed 
to be a classification of finitely many classes as decidable or undecidable, such that every standard class either contains 
an undecidable class (and hence is undecidable) or is contained in a decidable class (and hence is decidable).  This 
important feat was accomplished by the work of many researchers over many years, resulting in a classification into 
16 undecidable and 6 decidable classes, summarised on p.11 of~\cite{classical-book}.  Since in a class for 
which the decision problem is undecidable for single formulae, it will be \emph{a fortiori} undecidable
for theories, we will generally be interested in the decidable syntactic classes and subclasses thereof.

\myparagraph{Normal Forms}{
Note that when classifying formulae in this way it is obviously necessary to restrict to formulae in prenex form---that is, 
which have all their quantifiers brought to the front---since a formula of the form $\neg (\exists x.\varphi)$ is equivalent to 
$\forall x.\neg \varphi$, and so should be thought of as being universal rather than existential.  
Since we work with bounded variable logics where a variable may be quantified many times, we adopt the (slightly weaker) requirement that formulae be in \emph{negation normal form} (NNF).
We say that a 
formula $\varphi$ is `in NNF with quantifier prefix $s$, for $s\in \{\forall, \exists\}^*$, if (1) negations appear only in atomic formulae, 
and (2) $s$ is a topological sort of the quantifiers appearing in $\varphi$ (i.e. a linear ordering compatible with the partial order from the parse tree of $\varphi$).  
We say that a formula is `in NNF' if it is in NNF with quantifier prefix $s$ for some $s$.
Clearly, a formula in NNF with 
quantifier prefix $s$ can be put in prenex form with prefix sequence $s$. 
In \propref{prenex} we show that any regular theory can be transformed into another one with only NNF formulae.
}

}

\subsection{The decision problem for theories}

In this work, we investigate decision problems for \emph{theories}.
An FO \emph{theory} $T$ is simply a (finite or infinite) set of FO formulae.
For a theory $T$, we say that `$\bigwedge T$ is satisfiable' if
there is a model $\model$ such that
for every $\varphi \in T$, we have $\model \models \varphi$, i.e.,
 `$\model \models \bigwedge T$'.  
 We often refer to this problem as the 
\emph{conjunctive satsifiability} problem.  
We say that `$\bigvee T$ is satisfiable' if there is a model 
$\model$ such that there exists $\varphi \in T$ for which 
$\model \models \varphi$, i.e., `$\model \models \bigvee T$',
or equivalently, if there exists $\varphi \in T$ such that 
$\varphi$ is satisfiable.
We often refer to this problem as the \emph{disjunctive satisfiability} problem;
note that this is dual to (conjunctive) validity, 
since $\bigwedge T$ is valid if and only if $\bigvee\{\neg\varphi | \varphi \in T\}$ 
is unsatisfiable.


In order to obtain an algorithmic problem, 
we need that our (infinite) FO theories are finitely presented.
A naive finite presentation may be to specify a Turing machine and say 
that the theory will be the set of formulae accepted by the machine.  However, this model is  
so strong that everything is trivially undecidable: 
given a machine $M$, we can form the machine that accepts the
formula $\top$ padded to size $n$ if and only if $M$ halts in $n$ steps and the disjunction of the 
theory accepted by this machine will be satisfiable if and only if $M$ eventually halts; 
a similar construction can be used for the case of conjunctive satisfiability 
with a machine that outputs $\bot$.  
The conjunctive and 
disjunctive satisfiability problems are thus undecidable even 
when restricted to formulae containing only 
the literals $\top$ and $\bot$, plus a simple padding using conjunction with many copies of 
$\top$!
We will therefore consider a much weaker model for presenting our theories.
In particular, we will resort to theories that are recognised by a finite 
state machine, i.e. they are \emph{regular} languages.

\subsection{Regular theories}
\seclabel{regular-theories-def}

A set of logical formulae is \emph{regular} if it is generated by a regular grammar.  Such a grammar
consists of a set of non-terminals, which we divide into a set of `proposition' non-terminals 
$\{\Phi_1,\Phi_2,\ldots\}$ and a set of `term' non-terminals $\{X_1,X_2,\ldots\}$.  Each non-terminal
has a production rule, which for term non-terminals is an alternation of entries of the form $x$, $c$, or $f(Y_1,\ldots,Y_k)$, 
for variable $x$, constant $c$, function symbol $f$ of arity $k$ and term non-terminals 
$Y_1,\ldots,Y_k$.  For proposition non-terminals, the possible entries are of the form $\neg \Phi$, 
$\Phi \land \Phi'$, $\Phi \lor \Phi'$, $\exists x.\Phi$, $\forall x.\Phi$, $X_1=X_2$, or $R(X_1,\ldots,X_k)$,
for proposition non-terminals $\Phi$ and $\Phi'$, variable $x$, function symbol $f$ and relation symbol $R$
each of arity $k$, and term non-terminals $X_1,\ldots,X_k$.  We write $L(\Phi)$ (or $L(X)$) for the 
language generated starting at the non-terminal $\Phi$ (respectively $X$).
When writing grammars, we will allow ourselves some flexibility with notation, for instance to write
entire formulae in a single production rule entry; it is clear that this does not increase expressiveness.

As a simple example, the theory whose conjunction asserts that $y \neq f^k(x)$ for all $k\geq 0$ is $L(\Phi)$ according to the following grammar:
\begin{align*}
\Phi ::= \neg(y = X) \quad\quad\quad
X ::= x \ | \  f(X).
\end{align*}
Another theory enforcing the same thing is $\bigwedge L(\Psi)$ according to this grammar:
\begin{align*}
\Psi &::= \neg (x=y) \ | \ \exists x'. \left( x' = f(x) \wedge \exists x. \left(x=x' \wedge \Psi
\right) \right).
\end{align*}
\noindent
This technique, of using a dummy variable and existential quantifiers to `change' the value of $x$
to $f(x)$, is one we will use frequently later on.

Recall from our definition of negation normal form from Section \ref{sec:fologic}, that quantifiers appear only immediately before atomic propositions.  We next show that every regular theory can indeed
be transformed into NNF, and will subsequently assume that  all theories are in NNF.
\begin{proposition}[Negation normal form]\label{prop:prenex}
    Given a regular theory $T$, we can produce a regular theory $T'$ consisting only of formulae in negation normal form which is equivalent to $T$, in the sense that for any model $\model$ we have $\model \models \bigwedge T \Leftrightarrow \model \models \bigwedge T'$ and $\model \models \bigvee T \Leftrightarrow \model \models \bigvee T'$.
\end{proposition}
\newcommand{\negTrText}{\textsf{NEG}\xspace}
\newcommand{\negTr}[1]{\negTrText\big(#1\big)}
\begin{proof}
Suppose we are given a grammar $\mathcal{G}$ for $T$.
To obtain the desired NNF version $\mathcal{G}'$, we start with $\mathcal{G}$,
and first introduce a fresh propositional non-terminal $\Phi^c$,
for each propositional nonterminal $\Phi$ in $\mathcal{G}$,
and also replace every occurrence of $\neg \Phi$ in production rules of $\mathcal{G}$ with $\Phi^c$.
Next, we describe how to obtain the production rules for each of the newly 
introduced non-terminals of the form $\Phi^c$, using the production rules of
the nonterminal $\Phi$.
In particular, for each rule $\Phi \to \alpha$ in $\mathcal{G}$,
we add the rule $\Phi^c \to \negTr{\alpha}$, where the function $\negTrText$ is defined
as follows (all occurrences of $(\Psi^c)^c$ below must be replaced by $\Psi$):
\begin{align*}
\begin{array}{rclcrcl}
\negTr{X = X'} &=& \neg (X=X') 
& \quad \quad&
\negTr{R(X_1,\ldots,X_k)} &=& \neg R(X_1,\ldots,X_k) \\
\negTr{\Phi_1 \land \Phi_2} &=& \Phi_1^c \lor \Phi_2^c
& \quad \quad&
\negTr{\Phi_1 \lor \Phi_2} &=& \Phi_1^c \land \Phi_2^c \\
\negTr{\exists x.\Psi} &=& \forall x.\Psi^c
& \quad \quad&
\negTr{\forall x.\Psi} &=& \exists x.\Psi^c
\end{array}
\end{align*}
Let $T'$ be the language generated by the transformed grammar.
Then the elements of $T'$ are in NNF and in one-to-one correspondence with equivalent elements of $T$.
\end{proof}

\myparagraph{Tree automata}{
For algorithmic purposes, it will be convenient to consider the equivalent characterisation of regular
theories as those accepted by finite tree automata.  
Briefly, a non-deterministic bottom-up tree automata is a tuple 
$\Aa = (Q, \alphabet, \delta, F)$
where $Q$ is a finite set of states, $\alphabet$ is a finite 
ranked alphabet of maximum rank $k$ (i.e. an alphabet in which each symbol is equipped with an arity $r \leq k$), $F \subseteq Q$ is the set of final states,
and
$\delta = \bigcup\limits_{i=0}^k \delta_i$ is the transition relation, 
with $\delta_i \subseteq Q^i \times \Sigma_i \times Q$.  Intuitively, a transition rule $(q_1,\ldots, q_i, \sigma, q)\in\delta_i$ 
states that if the automaton can reach states $q_1,\ldots,q_i$ on the $i$ arguments of $\sigma$ then it can 
reach $q$ on $\sigma$.  A tree is \emph{accepted} if the automaton can reach a state $q\in F$ on the root. 
See \cite{comon2008tree} for a more detailed treatment of tree automata.
}


\section{Bounded model properties}
\seclabel{bounded}

In the setting of the classical decision problem, one of the most popular
(and certainly the most straightforward) 
methods for establishing decidability is via the \emph{finite model property}.
A class $\mathcal{C}$
of formulae has the finite model property if
for every formula $\varphi$ in $\mathcal{C}$, 
$\varphi$ is satisfiable iff it has a finite model.  
This is sufficient for decidability because of 
G{\"o}del's completeness theorem~\cite{god30} --- if $\phi$ is unsatisfiable then there is a 
proof of the validity of $\neg \phi$.
As a result, we get a decision procedure that simultaneously enumerates candidate finite models for
$\phi$ as well as candidate proofs for $\neg \phi$.  
Five of the seven maximal decidable classes in~\cite{classical-book} enjoy the 
finite model property.

In the setting of regular theories, however, 
the finite model property is not sufficient for decidability.
Indeed, in \secref{undec}, we will show that for two classes with the finite model property, 
the satisfiability problem is undecidable for regular theories of formulae in these classes.
For the conjunctive case this is because $\bigwedge T$ may not have a finite 
model even if each $\phi\in T$ does.
For the disjunctive case this happens because, when $\bigvee T$ is
unsatisfiable, the completeness theorem does not give us a proof of this fact.

In this section we will give a pair of stronger conditions, 
which we call the \emph{bounded model property}, weak and 
strong versions, which suffice for decidability for conjunctive and disjunctive satisfiability 
respectively, and, in some sense, are the analogues of the finite model property in our setting.  
The strengthening will be that the 
size of the finite model for a formula $\phi$ is bounded as a function of the signature $\text{sig}(\phi)$ of 
$\phi$, independent of the particular formula.

\begin{definition}[Bounded model property]
\deflabel{bounded-model-property}
A class $\mathcal{C}$ of formulae is said to have a 
\emph{weak bounded model property} if there is a function $f$ such 
that for every $\phi\in \mathcal{C}$, if $\phi$ is satisfiable 
then it has a model of size at most $f(\text{sig}(\phi))$.
Moreover, $\mathcal{C}$ is said to have the \emph{strong bounded model property} if $f$ is a computable function.
\end{definition}

We next state (\lemref{tree-regular-given-model}) the main ingredient 
towards the proof of our decidability result (Theorem~\ref{thm:boundedmodel}) for 
theories from classes with the above bounded model properties  --- the set of formulae 
satisfying a given finite model (over a fixed finite signature) is regular.
This can be used to decide if a given model satisfies the conjunction or disjunction
of a regular theory (Corollary \ref{cor:modelcheck}).
We remark that \lemref{tree-regular-given-model} is also observed in recent work on learning formulae~\cite{Krogmeier2022} to discriminate a finite set of finite models.

\begin{lemma}
\lemlabel{tree-regular-given-model}
Let $\vocab$ be a finite signature, $\vars$ be a finite set of variables,
and $\model = (U, \interp{})$ be a finite FO structure over $\vocab$.
The set of FO formulae $T_{\model} = \setpred{\varphi}{\model \models \varphi, \text{sig}(\varphi)\subseteq \vocab \uplus \vars}$ is tree regular.
\end{lemma}

\begin{proof}
\newcommand{\valstates}{Q_\textsf{val}}
\newcommand{\boolstates}{Q_\textsf{bool}}
We construct a tree automaton $\Aa_\model = (Q, \alphabet^{\vocab, \vars}, \delta, F)$ 
that accepts a tree $t_\varphi$
corresponding to a FO formula $\varphi$ iff $\model \models \varphi$.
At a high level, for each node $n$ of the tree $t_\varphi$ corresponding to
a (sub-)term $t_n$ in $\varphi$, the automaton annotates $n$ with the set of all
pairs $(e, \gamma)$ such that the term $t_n$ will evaluate to element $e \in U$ 
under the assignment $\gamma$.
Likewise, for each node $n$ of the tree $t_\varphi$ corresponding to a 
sub-formula $\psi$, the automaton annotates $n$ with the set 
of assignments under which
$\psi$ evaluates to $\tru$ (or alternatively, those assignments under which it evaluates to $\fls$).
The states of this automaton essentially consist of the set of all such annotations,
and are finitely many, since both $\vocab$ and $\vars$ are finite.
The transitions of the automaton outline how one can 
infer the annotations of a node from the annotations of its children.

The set of states of $\Aa_\model$ is the disjoint union 
$Q = \valstates \uplus \boolstates$
where $\valstates$ and $\boolstates$ are defined as follows.
Let $\mathsf{Asgns} = [\vars \to U]$ be the set of all functions 
from $\vars$ to $U$. 
Then, $\valstates = U \times \powset{\mathsf{Asgns}}$ and 
$\boolstates = \set{\tru, \fls} \times \powset{\mathsf{Asgns}}$.
The set of final states is $F = \setpred{(\tru, \Gamma)}{\mathsf{Asgns} \supseteq \Gamma \neq \emptyset}$.
In the following, we formally describe the transitions $\delta_\sigma$ 
for each symbol $\sigma \in \alphabet^{\vocab, \vars}$ separately;
$\delta$ will then simply be $\bigcup_{\sigma \in \alphabet^{\vocab, \vars}} \delta_\sigma$.
\begin{description}
	\item[\textbf{Case $\sigma \in \vars$.}] 
	In this case, the run of the automaton
	picks a state of the form $q  = (e, \Gamma) \in \valstates$ so that, under all assignments
	in $\Gamma$, $\sigma$ evaluates to the element $e \in U$. That is,
	\begin{align*}
	\delta_\sigma = \setpred{\big(\sigma, (e, \Gamma)\big)}{e \in U, \Gamma \subseteq \setpred{\gamma \in \mathsf{Asgns}}{\gamma(\sigma) = e}}
	\end{align*}

	\item[\textbf{Case $\sigma \in \consts \cup \funcs \cup \rels$ with $\arity{\sigma} = r \geq 0$}.]
	Here, in the case of $\sigma \in \consts$, the transitions are similar
	to the case when $\sigma \in \vars$.
	In the case of $\sigma \in \funcs$, the automaton 
	chooses a state for the current node based on the states of the children.
	In particular, if the states of the children are respectively $(e_1, \Gamma_1) \ldots, (e_r, \Gamma_r)$,
	then the automaton picks a state $(e, \Gamma)$ such that every assignment $\gamma \in \Gamma$
	is also in each of $\set{\Gamma_i}_i$ and further, the tree rooted at the current
	node evaluates to $e$ under every $\gamma \in \Gamma$.
	Such a choice for the state ensures consistency with the states of the children.
	The case of $\sigma \in \rels$ is similar.
	The precise formal details are as follows.
	\begin{align*}
	\delta_\sigma = \setpred{\big((e_1, \Gamma_1), \ldots, (e_r, \Gamma_r), \sigma, (e, \Gamma)\big)}{e \in U \uplus \set{\tru, \fls},
	\Gamma \subseteq \bigcap_{i=1}^r \Gamma_i, e = \interp{\sigma}(e_1, \ldots, e_r)}
	\end{align*}
	
	\item \textbf{Case $\sigma$ is some $\oplus \in \set{\neg, \land, \lor}$ with $\arity{\oplus} = r$} 
	In this case, the automaton picks a state whose result
	can be obtained by the boolean combination (under $\oplus$) of the results of the children, and the assignments are chosen to be a common subset of those of the children. That is,
	\begin{align*}
	\delta_\sigma = \setpred{\big((b_1, \Gamma_1), \ldots, (b_r, \Gamma_r), \sigma, (b, \Gamma)\big)}{b = \oplus(b_1, \ldots, b_r),
	\Gamma \subseteq \bigcap_{i=1}^r \Gamma_i}
	\end{align*}

	\item \textbf{Case $\sigma$ is `${=}$'.} 
	Here, the automaton picks a state whose first component is $\tru$ whenever the
	values on both children are picked to be the same. That is,
	\begin{align*}
	\delta_\sigma = \setpred{((e_1, \Gamma_1), (e_2, \Gamma_2), \text{`{=}'}, (b, \Gamma)\big)\!\!}{\!\!e_1, e_2 \in U, \Gamma \subseteq \Gamma_1\cap \Gamma_2, 
	 b \text{ is } \tru \text{ if } e_1 = e_2 \text{ and } \fls \text{ otherwise}
	}
	\end{align*}

	\item \textbf{Case $\sigma$ is `$\exists x$' for some $x \in \vars$.} 
	This is the interesting case. Here, the automaton can pick a state $(\tru, \Gamma)$
	such that $\Gamma$ contains the assignments, derived from those of the children,
	by replacing $x$ with an arbitrary element from the universe.
	Alternatively, the automaton can pick a state whose result is $\fls$ while
	the set $\Gamma$ of assignments picked is obtained by considering only
	those assignments from that of the child formula, which return $\fls$,
	no matter how the element corresponding to $x$ is chosen.
	Both these transitions closely mimic the semantics of the existential quantifier.
	Formally,
	\begin{align*}
		\delta_\sigma = \cup
		\begin{array}{l}
		\setpred{\big((\tru, \Gamma_1), \sigma, (\tru, \Gamma)\big)}{
		\Gamma \subseteq \setpred{\gamma[x \mapsto e] }{ e \in U, \gamma \in \Gamma_1}
		} \\
		\setpred{\big((\fls, \Gamma_1), \sigma, (\fls, \Gamma)\big)}{
		\Gamma \subseteq \setpred{\gamma \in \mathsf{Asgns}}{\forall e\in U, \gamma[x \mapsto e] \in \Gamma_1}
		}
		\end{array}
	\end{align*}

	\item \textbf{Case $\sigma$ is `$\forall x$' for some $x \in \vars$.} 
	This is the dual of the previous case and follows similar reasoning.
	\begin{align*}
		\delta_\sigma = \cup
		\begin{array}{l}
		\setpred{\big((\tru, \Gamma_1), \sigma, (\tru, \Gamma)\big)}{
		\Gamma \subseteq \setpred{\gamma \in \mathsf{Asgns}}{\forall e\in U, \gamma[x \mapsto e] \in \Gamma_1}
		} \\
		\setpred{\big((\fls, \Gamma_1), \sigma, (\fls, \Gamma)\big)}{
		\Gamma \subseteq \setpred{\gamma[x \mapsto e] }{ e \in U, \gamma \in \Gamma_1}
		}
		\end{array}
	\end{align*}
\end{description}
The proof of correctness of the above construction follows from a simple structural induction.
\end{proof}

For a given regular theory $T$, the set of variables, 
relation and function symbols appearing in the grammar generating $T$ is finite, 
and hence $T$ can be considered to be over a finite signature $\vocab$ and finite
set of variables $\vars$.  \lemref{tree-regular-given-model} thus immediately gives:

\begin{corollary}
\corlabel{modelcheck}
The model-checking problems of determining for a given regular theory $T$ and finite model $\model$ whether 
$\model \models \bigwedge T$ and whether $\model \models \bigvee T$, is decidable.
\end{corollary}



Equipped with \lemref{tree-regular-given-model}
and Corollary~\ref{cor:modelcheck},
we can now prove the devidability results for classes with 
bounded model properties (\defref{bounded-model-property}).

\boundedModelProperty*

\noindent
\begin{proof}
We again observe that for any regular theory $T$, all formulae 
of $T$ are over a fixed finite set of symbols from the signature $\vocab$ and variables $\vars$ 
appearing in the grammar generating $T$. 
Then, to decide whether $\bigvee T$ 
is satisfiable, we can just check all models up to size $f(\vocab \cup \vars)$, which can be computed
for a theory with strong bounded model property.  

We now turn our attention to conjunctive satisfiability when the underlying class only has weak bounded model property.
If $\bigwedge T$ is satisfiable, then it has a model of size at most $f(\vocab \cup \vars)$ 
(although we cannot necessarily compute it). 
On the other hand if it is unsatisfiable, then by the compactness theorem 
for first-order logic there is a finite subset $T' \subseteq T$ such that 
$\bigwedge T'$ is unsatisfiable, and by G{\"o}del's completeness theorem there is a proof of this.  
The decision procedure now uses the classic ``dove-tailing'' technique,
where one interleaves multiple countable enumerations and terminates 
when one of them terminates.
Here, we enumerate (1) finite models that witness the satisfiability of the set $T$, 
and interleave it with the enumeration of 
(2.1) (an increasing chain of) finite subsets of the theory $T$, together with
(2.2) bounded length proofs of unsatisfiability of each such finite subset. 
More precisely, let $\varphi_1, \varphi_2, \ldots$ be a 
computable enumeration of $T$. 
Likewise let $\model_1, \model_2, \ldots$ be a 
computable enumeration of finite models; we note that the 
collection of finite models is countable and admits a computable enumeration. 
Then, at step $i$ of our decision procedure, we:
\begin{enumerate}
\item Check if the model $\mathcal{M}_i$ satisfies $T$. This is a decidable check (see Corollary~\ref{cor:modelcheck}). If so, we conclude that $\bigwedge T$ is satisfiable.
\item Enumerate all natural deduction proofs of length $\leq i$ for the unsatisfiability of the formula 
$\bigwedge_{j=1}^i \varphi_j$. If there is one, we conclude that $\bigwedge T$ is unsatisfiable
\end{enumerate}
The proof of correctness of the above is as follows.
If $\bigwedge T$ is satisfiable, there is a finite model that witnesses 
the satisfiability (by the weak bounded model property), 
and the above procedure will terminate. 
On the other hand, if $\bigwedge T$ is unsatisfiable, 
then there will be a finite subset $T’$ of it which is unsatisfiable. 
Let $i$ be the smallest index so that 
$T’ \subseteq S_i = \set{\varphi_1, \varphi_2, \ldots \varphi_i}$. 
Clearly $S_i$ is unsatisfiable, and thus has 
a natural deduction proof of unsatisfiability, 
whose length is, say, $j$. Let $k = max(i, j)$. 
Observe that there is a proof of length $\leq k$ 
for the unsatisfiability of the set $S_k = \set{\varphi_1, \varphi_2, \ldots, \varphi_k}$.
Thus, again, our procedure terminates (in $k$ steps) with the correct answer.
\end{proof}



\section{Undecidable classes}\label{sec:undec}

In this section, we will consider two syntactic classes for which the classical decision problem is decidable,
and show that their decision problems for regular theories are undecidable.  The two classes are the \emph{EPR} 
(or \emph{Bernays-Sch{\"o}nfinkel}) class $\epr$, shown decidable by Ramsey~\cite{ramsey1930problem}, and the 
\emph{Gurevich} class $\gurevich$, shown decidable by Gurevich~\cite{gurevich1976decision}. In both cases decidability is established through
the finite model property.

For each of these classes, we show that there is a subclass, respectively 
$\eprweak$ (consisting of formulae obtained by a
conjunction of a purely existential formula and a purely universal formula with at most $3$ universally quantified variables) and $\gurevichweak$, for which the 
conjunctive and disjunctive satisfiability problems are undecidable\footnote{Note that $\eprweak$ is not a 
standard syntactic class, but the fact that it is undecidable shows that the EPR class (which contains it) is undecidable.  More precisely it shows that the standard class $[\exists^*\forall^3,(0,2),(0)]_=$  
(contained in the EPR class) is undecidable.  Note that the other standard classes arising from other 
interleavings of $\forall^3$ with $\exists^*$ are already undecidable for single formulae, since they each contain one of the classicaly undecidable classes $[\forall^3\exists^*,(0,1),(0)]$ and $[\forall\exists^*\forall,(0,1),(0)]$.}.


\subsection{The EPR class}

\eprundec* 

\begin{proof}
The proof proceeds by reduction from the 
\emph{Post Correspondence Problem} (PCP), a classical undecidable problem~\cite{Post1946Correspondence}. 
An instance of PCP is a set of pairs of strings $\set{(u_1,v_1),\ldots,(u_k,v_k)}$, 
with each $u_i,v_i\in \set{0,1}^*$,
and the task is to determine whether there exists a sequence of 
indices $i_1,\ldots,i_n$ such that the concatenations 
$u_{i_1}\ldots u_{i_n}$ and $v_{i_1}\ldots v_{i_n}$ are equal.

\begin{figure*}
\begin{align*}
\Psi &::= \left((\phi_0 \wedge (x = y)\right) \wedge 
\Phi \\
\phi_0 &= \forall y'. (\neg s_0(y',x) \wedge \neg s_1(y',x)) \wedge \forall x. \forall y. \left(\bigwedge_{i\in\{0,1\}} (s_i(x,y) \wedge s_i(x,y')) 
\Rightarrow y=y'\right) \\ &\qquad \wedge \left(\bigwedge_{i,j\in\{0,1\}} (s_i(y,x) \wedge s_j(y',x)) 
    \Rightarrow (y=y' \wedge i=j)\right), \\
\Phi &::= (x=y) \ | \ \Big|_{i=1}^k \Psi(u_i,v_i) \\
\Psi(u_i,v_i) &::= \exists x'.\left( s_{u_i(1)}(x,x') \wedge \exists x. \left(x = x' \wedge \exists x'.\left( s_{u_i(2)}(x,x') \wedge \exists x.\left( x=x' \wedge \ldots \right.\right.\right.\right.\\
&\qquad \wedge \exists x'. \left(s_{u_i(k_i)}(x,x') \wedge \exists x. \left(x=x' \wedge \exists y'. \left(s_{v_i(1)}(y,y') \wedge \exists y.\left( y=y' \wedge \ldots \right.\right.\right.\right.\\
&\qquad \left.\wedge 
    \exists y'. \left(s_{v_i(l_i)}(y,y') \wedge \exists y. \left(y=y' \wedge \Phi\right)\right)\ldots\right).
\end{align*}
\caption{Logical encoding of the Post Correspondence Problem (disjunctive version)}\label{fig:eprforms}
\begin{align*}
\Psi' &::= \left((\phi_0 \wedge (x = y)\right) \wedge 
    \Phi',\\
\Phi' &::= (x\neq y) \ | \ \Big|_{i=1}^k \Psi(u_i,v_i) \\
\Psi'(u_i,v_i) &::= \exists x'.\left( s_{u_i(1)}(x,x') \wedge \exists x. \left(x = x' \wedge \exists x'.\left( s_{u_i(2)}(x,x') \wedge \exists x.\left( x=x' \wedge \ldots \right.\right.\right.\right. \\
&\qquad \left.\left.\left.\left. \wedge \exists x'. \left(s_{u_i(k_i)}(x,x') \wedge \exists x. \left(x=x' \wedge \right.\right.\right.\right. 
    \exists y'. \left(s_{v_i(1)}(y,y') \wedge \exists y.\left( y=y' \wedge \ldots \right.\right.\right.\right.\\
    &\qquad \left.\wedge 
    \exists y'. \left(s_{v_i(l_i)}(y,y') \wedge \exists y. \left(y=y' \wedge \Phi'\right)\right)\ldots\right)
\end{align*}
\caption{Logical encoding of the Post Correspondence Problem (conjunctive version)}\label{fig:eprformscon}
\end{figure*}

\myparagraph{Disjunctive satisfiability}
For disjunctive satisfiability, consider the theory $L(\Psi)$ as shown in Figure \ref{fig:eprforms}, which is over 
a signature of four variables and two binary relations $s_0$ and $s_1$, with equality.  We claim 
that $\bigvee L(\Psi)$ is satisfiable if and only if the given PCP instance has a solution.

Indeed, suppose that some formula $\phi \in L(\Psi)$ is satisfiable.  Note that the second part of $\phi_0$ 
enforces that each $s_i$ is functional, in that for each $x$ there 
is at most one $y$ such that $s_i(x,y)$.  Writing $s_i(x)$ for such a $y$ if it exists, the third part 
enforces that $s_i(x)=s_j(y)$ only if $x=y$ and $i=j$, and the first part that the initial value of $x$ is not in the 
image of $s_0$ or $s_1$.  Hence inductively we have that 
$s_{i_k}(s_{i_{k-1}}(\ldots s_{i_1}(x))\ldots ) = s_{j_l}(s_{j_{l-1}}(\ldots s_{j_1}(x))\ldots )$ if and only if $k=l$ and $i_m=j_m$ for all $m$.

\renewcommand{\vec}[1]{\mathbf{#1}}

Observe that a formula $\phi$ of $L(\Psi)$ is formed by a series of choices $i_1,\ldots,i_n$ for $i$, followed by a final $x=y$.
The formula obtained from a single production of $\Psi(u_i,v_i)$ takes the values of $x$ and $y$ to $s_{u_i(k_i)}(s_{u_i(k_i-1)}(\ldots s_{u_i(1)}(x))\ldots )$ 
and $s_{v_i(l_i)}(s_{v_i(l_i-1)}(\ldots s_{v_i(1)}(y))\ldots )$ respectively; denote these as $s_{\vec{u_i}}(x)$ and 
$s_{\vec{v_i}}(y)$.  Hence at the end of the formula, we have that the final values of $x$ and $y$ are 
$s_{\vec{u_{i_n}}}(s_{\vec{u_{i_{n-1}}}} ( \ldots s_{\vec{u_{i_1}}}(x_0))\ldots )$ and \\ 
$s_{\vec{v_{i_n}}}(s_{\vec{v_{i_{n-1}}}} ( \ldots s_{\vec{v_{i_1}}}(y_0))\ldots )$
respectively, where $x_0=y_0$ are the initial values of $x$ and $y$.  Hence $\phi$ is satisfiable only 
if $i_1,\ldots,i_n$ is a solution to the given PCP instance.  \sloppy

Conversely, if $i_1,\ldots,i_n$ is a solution to the given PCP instance then the corresponding formula $\phi$ 
holds on the infinite binary tree, where $s_0$ and $s_1$ are interpreted as the left and right child relations respectively,
and $x$ and $y$ initially interpreted as the root.

\myparagraph{Conjunctive satisfiability}
For conjunctive satisfiability, we claim that $\wedge L(\Psi')$ is satisfiable if and only if the given PCP instance does 
\emph{not} have a solution ($\Psi'$ as shown in Figure \ref{fig:eprformscon}; the same as $\Psi$ but with the final $x=y$ replaced 
by $x\neq y$).

Indeed, if the given instance has no solution then the infinite binary tree, with $s_0$ and $s_1$ interpreted as the left and 
right child relations, and $x$ and $y$ initially interpreted as the root, is a model for $\wedge L(\Psi')$.  Conversely, if $\wedge L(\Psi')$ 
is satisfiable then every formula $\phi\in L(\Psi')$, say arising from the sequence $i_1,\ldots,i_n$, satisfies $s_{\vec{u_{i_n}}}(s_{\vec{u_{i_{n-1}}}} ( \ldots s_{\vec{u_{i_1}}}(x_0))\ldots ) \neq 
s_{\vec{v_{i_n}}}(s_{\vec{v_{i_{n-1}}}} ( \ldots s_{\vec{v_{i_1}}}(y_0))\ldots )$ and hence (recalling that $x_0=y_0$) we must have 
$u_{i_1}u_{i_2}\ldots u_{i_n} \neq v_{i_1}v_{i_2}\ldots v_{i_n}$ and so $i_1,\ldots,i_n$ is not a solution of the given PCP instance,
as required.
\end{proof}

Note that the spirit of this proof is a kind of inverse of Rabin's theorem that MSO on trees is 
decidable~\cite{rabin1969decidability}---the universal parts of $L(\Psi)$ and $L(\Psi')$ are enforcing that the model is
sufficiently close to a binary tree.


\subsection{Existential formulae with unary functions}

In this section, we show that the satisfiability problem for regular theories
over the purely existential fragment of FO is undecidable, even for a signature limited to two unary functions.



\exundec*


\begin{proof}
\newcommand{\halttext}{\textsf{HALT}}
\newcommand{\start}{q_{\textsf{START}}}
\newcommand{\halt}{q_{\halttext}}
\newcommand{\inc}{\mathsf{INC}}
\newcommand{\dec}{\mathsf{DEC}}
\newcommand{\chk}{\mathsf{CHK}}
\newcommand{\ctr}{\mathsf{ctr}}
\newcommand{\zr}{\texttt{0}}
\newcommand{\scr}{\mathsf{s}}
\newcommand{\pdr}{\mathsf{p}}

Our proof is via a reduction from the halting problem for $2$-counter Minsky machines --- 
given a $2$-counter machine $M$, we will construct a regular set of existentially quantified
FO formulae $T$
such that $\bigvee T$ is satisfiable if and only if $M$ can halt, and similarly a set $T'$ such that $\bigwedge T'$ 
is satisfiable if and only if $M$ does not halt.

We first recall the formal definition of a $2$-counter Minsky machine.
Such a machine is a tuple $M = (Q, \start, \halt, \ctr_A, \ctr_B, D, \delta)$,
where $Q$ is a finite set of states, $\start \neq \halt \in Q$ are designated
initial and halting states respectively, $\ctr_A$ and $\ctr_B$ are counters
that take values over $\nats = \set{0, 1, \ldots}$
and $\delta \subseteq Q \times \set{A, B} \times \set{\inc,\dec,\chk} \times Q$
is the transition relation.
A \emph{run} of such a machine is a finite non-empty sequence $\rho \in (\delta \times (\nats \times \nats))^+$:
\begin{align*}
\rho = &\Big((q_1, C_1, \sigma_1, q'_1), (v^A_1, v^B_2)\Big), \Big((q_2, C_2, \sigma_2, q'_2), (v^A_2, v^B_2)\Big) \ldots \\
 \ldots &\Big((q_{k-1}, C_{k-1}, \sigma_{k-1}, q'_{k-1}), (v^A_{k-1}, v^B_{k-1})\Big), \Big((q_k, C_k, \sigma_k, q'_k), (v^A_k, v^B_k)\Big)
\end{align*}
such that $q_1 = \start$ and for every $2 \leq i \leq k$, $q_i = q'_{i-1}$.  We say $\rho$ is \emph{feasible} if 
for every  $1 \leq i \leq k$ one of the following holds:
\begin{enumerate}[label=(\alph*)]
\item $v^{C_i}_i = v^{C_i}_{i-1} + 1$, $\sigma_i = \inc$, and $v^{\overline{C_i}}_i = v^{\overline{C_i}}_{i-1}$ 
\item $v^{C_i}_i = v^{C_i}_{i-1} - 1 \geq 0$, $\sigma_i = \dec$, and $v^{\overline{C_i}}_i = v^{\overline{C_i}}_{i-1}$
\item $v^C_i = v^C_{i-1}$ for each $C \in \set{A, B}$, $v^{C_i}_i=0$ and $\sigma_i = \chk$,
\end{enumerate}
writing $v^A_0 = v^B_0 = 0$
and using the notation $\overline{A} = B$ and $\overline{B} = A$.
We say $\rho$ reaches state $q \in Q$ if  $q'_k = q$.

The halting problem asks whether a given $2$-counter machines $M$ has a feasible run that reaches $\halt$,
and is undecidable in general.

\newcommand{\rewrite}{\textsf{copy}}
\newcommand{\unequalconsts}{\textsf{consts}}
\newcommand{\scrpdr}{\textsf{s,p}}


\begin{figure*}
\begin{align*}
\begin{array}{rcl}
\Phi &::=& y_A = c_\zr \land y_B = c_\zr \land \Psi_{\start} \\
\Psi_q &::=& \text{ if $q=\halt$ then $\top$ else $\bot$}\\
&& \big|_{(q,C,\inc,q') \in \delta} \exists y'_A, y'_B \Big( \psi_{\inc,C} \land \psi_\scrpdr 
 \land \big(\exists  y_A, y_B \, (\psi_\rewrite \land \Psi_{q'})\big)\Big) \\
&& \big|_{(q,C,\dec,q') \in \delta} \exists y'_A, y'_B \Big( \psi_{\dec,C} \land \psi_\scrpdr 
\land \neg(y_C = c_\zr) \land \big(\exists  y_A, y_B \, (\psi_\rewrite \land \Psi_{q'})\big)\Big) \\
&& \big|_{(q,C,\chk,q') \in \delta} \exists y'_A, y'_B \Big( \psi_\rewrite \land \psi_\scrpdr 
\land y_C = c_\zr \land \big(\exists  y_A, y_B \, (\psi_\rewrite \land \Psi_{q'})\big)\Big) \\
\psi_\rewrite &::=& y_A = y'_A \land y_B = y'_B \\
\psi_\scrpdr &::=& \revise{\neg(\scr(y'_A) = y'_A) \land \neg(\pdr(y'_A) = y'_A)} \\
&& \land \pdr(\scr(y'_A)) = y'_A \land \scr(\pdr(y'_A)) = y'_A 
 \land \pdr(\scr(y'_B)) = y'_B \land \scr(\pdr(y'_B)) = y'_B \\
\psi_{\inc,C} &::=& y'_{\overline{C}} = y_{\overline{C}} \land y'_C = \scr(y_C) \\
\psi_{\dec,C} &::=& y'_{\overline{C}} = y_{\overline{C}} \land y'_C = \pdr(y_C) \\
\end{array}
\end{align*}
\caption{Logical encoding of a 2-counter machine (disjunctive version)}
\figlabel{2countdis}

\input{exists-undec-conjunctive.tex}
\end{figure*}

\myparagraph{Disjunctive satisfiability}
Given a $2$-counter machine $M$, we present in  \figref{2countdis} a regular language $L(\Phi)$
that will contain a formula $\varphi_\rho$ for each run $\rho$, which is satisfiable 
if and only if $\rho$ is a feasible run that reaches $\halt$.

The vocabulary for $L(\Phi)$ is $\vocab = (\consts, \rels,\funcs)$,
where 
$\consts = \set{c_\zr}$, 
$\rels = \emptyset$ and 
$\funcs = \set{\scr, \pdr}$, with $\arity{\scr} = \arity{\pdr} = 1$
and the set of variables is $\vars = \set{y_A, y'_A, y_B, y'_B}$.
Recall that for $C \in \set{A, B}$, 
we use $\overline{C}$ to denote the other value in $\set{A, B}$.
Observe that each formula only contains existential quantifiers,
and all the negations appear in atomic formulae.

We now argue for correctness.
First assume that the given counter machine $M$ halts.
In this case, there is a run $\rho$ that ends in $\halt$.
Then the formula $\varphi_\rho$ corresponding to $\rho$ will be satisfiable
on the model of the natural numbers, with $\scr(\cdot)$, $\pdr(\cdot)$ interpreted as the successor
and predecessor functions respectively, and $c_\zr$ interpreted as $0 \in \nats$.

For the other direction, suppose that $\varphi \in L(\Phi)$ is a satisfiable formula, 
over some model $\MM$.
From the parse tree of $\varphi$ construct a run $\rho$ by at each step choosing the 
transition corresponding to the one used in the production of $\phi$.  
Then this is an 
accepting run: indeed, the occurrences of $\psi_{s,p}$ in $\varphi$ 
enforce that the terms $c_0,\scr(c_0),\scr^2(c_0),\ldots$ (as far as they correspond to elements
corresponding to terms in $\varphi$) map to distinct elements of the universe, 
and further, the functions $\scr$ and $\pdr$ behave as inverses
when restricted to this set of elements.  
This means that $y_A$ and $y_B$ do indeed behave as
counters and so the conditions $y_C \neq c_0$ (for $\dec$ transitions) and $y_C=c_0$ (for 
$\chk$ transitions) ensure that $\rho$ is feasible.  Moreover for $\phi$ to be satisfiable 
the final leaf must be $\top$ rather than $\bot$ and so $\rho$ must end in $q$.

\myparagraph{Conjunctive satisfiability}{
This time we reduce from the problem of checking if a given 2-counter
machine $M$ does not halt.
The grammar for describing $L(\Phi')$, the regular set of formulae,
is as shown in \figref{2countcon} over the same vocabulary $\vocab$.
Here, each formula corresponds to a run $\rho$ and intuitively encodes
that `$\rho$ is either infeasible or does not end in $\halt$'.

The argument for correctness is similar to the disjunctive case.  If $M$ does not halt then 
the natural numbers are a model for $\bigwedge L(\Phi')$: for any run $\rho$ either $\rho$ does 
not end in $\halt$ in which case the final leaf is $\top$ (and this propagates up through the 
$\vee$ nodes), or there is a $\dec$ or $\chk$ condition which fails and so the corresponding 
conditions $y_C=c_0$ or $y_C \neq c_0$ respectively are true (and this again propagates up to 
the root node).

Conversely, we claim that if $\phi_\rho$ is satisfiable then $\rho$ is not a feasible run 
ending in $\halt$.  If the final leaf of $\phi_\rho$ is $\top$ then $\rho$ does not end in 
$\halt$.  On the other hand if the final leaf is $\bot$ then there must be a $\dec$ or $\chk$ 
transition for which the corresponding $y_C=c_0$ or $y_C \neq c_0$ conditions hold, and up to 
this point we again have that the $\psi_{s,p}$ conditions enforce that $y_A$ and $y_B$ behave 
as counters and so the corresponding check is indeed failed in $\rho$ so $\rho$ is not feasible.
}
\end{proof}



\section{Decidable syntactic classes}

In this section we will show that both the EPR and the Gurevich classes contain subclasses which are decidable in our setting.  We will first show in Section \ref{sec:exun} that the classes $[\exists^*,\allc,(0)]_=$ and $[\forall^*,\allc,(0)]_=$ of (respectively) purely existential and purely universal formulae with no functions (which are subclasses of the EPR class, and in the case of the first also of the Gurevich class) are decidable for both conjunctive and disjunctive satisfiability.  We will then show in Section \ref{sec:euf} that the class $[\nonec,\allc,\allc]_=$ of quantifier-free formulae (which is a subclass of the Gurevich class) is decidable for conjunctive satisfiability.  We conjecture that disjunctive satisfiability is also decidable for this class, but leave this as open for future work.


\subsection{Purely existential and purely universal function-free theories}\label{sec:exun}

We will now show that the 
classes of purely existential and purely universal function-free theories, $[\exists^*,\allc,(0)]_=$ and $[\forall^*,\allc,(0)]_=$, 
have the strong bounded model property, and hence the conjunctive and disjunctive satisfiability problems are decidable for these
classes.  The purely universal case follows immediately from the fact that if a purely universal formula is satisfiable then it is 
satisfiable on 
\revise{a structure whose universe is the Herbrand universe,
comprising the set of all ground terms (or, in the presence of equalities, some quotient of this set corresponding to equivalence classes of the congruence induced by equality).  In particular, for a given function-free signature this has (computably) bounded size and so the strong bounded model property holds.
}

For the purely existential case, as observed in \cite{rosen1995preservation} (Proposition 11), for a given function-free signature $\vocab$ 
with $k$ variables there is a satisfiable finite theory $\Gamma_k$, the \emph{$k$-Gaifman theory}, \revise{depending only on $\vocab$,} such that for any existential formula 
$\phi$ over $\vocab$, if $\phi$ is satisfiable then $\Gamma_k \models \phi$. 
Moreover, $\Gamma_k$ can be straightforwardly computed for a given $\vocab$ 
(see~\cite[paragraph before Proposition 11]{rosen1995preservation}).
Therefore if we find a model for $\Gamma_k$, the size of this 
model is a bound on the size of model required for any existential formula over $\sigma$, and hence $[\exists^*,\allc,(0)]_=$ has the 
strong bounded model property.  We thus have:

\funcfree*


\newcommand{\congclos}{\textrm{cclos}}
\newcommand{\kcc}[1]{#1\text{-}\congclos}
\newcommand{\grow}{\textrm{grow}}
\newcommand{\subst}{\textrm{subst}}
\newcommand{\qap}{Q_{\text{ap}}}

\subsection{Conjunctive satisfiability of EUF formulae}
\seclabel{euf}

In this section we will consider the case of the class $[\nonec,\allc,\allc]_=$ of quantifier-free formulae, also known as EUF (equality logic with uninterpreted functions).  We will show that if $T$ is a regular set of quantifier-free formulae specified by a tree automaton then the problem of satisfiability of $\bigwedge T$ is decidable.  For single formulae this was shown to be decidable by Ackermann in 1954~\cite{ackermann1954solvable}, but modern algorithms are based on \emph{congruence closure}~\cite{shostak1978algorithm}.
For the rest of the section, we will assume that the signature is algebraic, i.e., $\rels = \emptyset$.
This is because, we can systematically replace each $k$-ary relation symbol $R$ with a fresh
function symbol $f_R$, introduce constants $c_\top \neq c_\bot$ and convert each occurrence of $R(t_1,\ldots,t_k)$ to $f_R(t_1,\ldots,t_k) = c_\top$ (and similarly $\neg R(t_1,\ldots,t_k)$ to $f_R(t_1,\ldots,t_k)=c_\bot$).
We remark that this transformation also preserves regularity.
We now recall congruence closure, an essential ingredient in our decidability result.

\begin{definition}
Let $E=\set{s_i=t_i}_{i\in I}$ be a set of equality atomic formulae.
The congruence induced by $E$,  is the smallest equivalence relation $R$
on $\terms{\vocab,\vars}$
 such that
\begin{itemize}
    \item $(s_i,t_i) \in R$ for all $s_i=t_i\in E$, and
    \item $R$ is a \emph{congruence}: for any function symbol $f$ of arity $k$ and terms $a_1,\ldots,a_k,b_1,\ldots,b_k$ with $(a_i,b_i)\in R$ for all $i$, we have $(f(a_1,\ldots,a_k),f(b_1,\ldots,b_k))\in R$.
\end{itemize}
Let $S\subseteq \terms{\vocab, \vars}$ be a set of terms.  
The \emph{congruence closure of $S$ under $E$}, $\congclos_E(S)$, is the set of terms related by $R$ 
to elements of $S$.
\end{definition}

The reason why congruence closure is useful for EUF is that it is sufficient to tell us whether a set of EUF constraints is satisfiable. 

\begin{proposition}\label{prop:congsound}
Let $T=E\cup D$, where $E$ and $D$ are sets of equality and disequality atomic propositions respectively.  Then $\bigwedge T$ is unsatisfiable if and only if there exists a disequality $t\neq t' \in D$ such that $t' \in \congclos_E(\{t\})$.
\end{proposition}

The main theorem we will need is that if $S$ is a regular set of terms then the congruence closure of $S$ under equating the elements of regular sets $S_1,\ldots,S_n$ is itself regular, and is computable as a function of the inputs.

\begin{theorem}\label{thm:congclos}
Let $\Aa,\Aa_1,\ldots,\Aa_n$ be tree automata accepting terms from $\terms{\vocab, \vars}$.  
Let $E=\setpred{s=t}{s,t\in L(\Aa_i) \text{ for some $i$}}$.  Then $\congclos_E(L(\Aa))$ is regular, and given by an automaton computable as a function of $\Aa,\Aa_1,\ldots,\Aa_n$.
\end{theorem}

This is essentially proved in \cite{dauchet1990decidability} (see final paragraph of p.195), but for 
completeness we provide a self-contained proof in 
Appendix \ref{sec:congreg}.

An additional piece of notation we use in
 Appendix \ref{sec:congreg} 
and later in this section is the following: given $n$ fresh symbols $X_1, X_2, \ldots, X_n$, a \emph{context} $C(X_1, \ldots, X_n)$
(or simply $C$)
over $\alphabet$ is a tree over $\alphabet$ enriched with arity-0 symbols $X_1,\ldots,X_n$,
with the restriction
that $C(X_1, \ldots, X_n)$ has exactly one occurrence of each of $X_1, X_2, \ldots, X_n$.
Given a context $C(X_1, X_2, \ldots, X_n)$ and terms $t_1, t_2, \ldots, t_n$
over the alphabet $\alphabet$, the expression $C[t_1/X_1, \ldots, t_2/X_2]$
is the term obtained by replacing $X_i$ by $t_i$ in the context $C$.
If for automaton $\Aa$ we have that starting in state $q_i$ at each $X_i$ the automaton can reach 
state $q$ at the root of $C$ then we write $q_1(X_1), \ldots,q_k(X_k) \rightarrow q(C)$.

We are now ready for the main theorem of this section.





\quantfreecon*
\noindent
The proof of \thmref{quantfreecon} proceeds as follows.  
We first consider the special case where the theory $T$ is of the form $\{s=t|s,t\in S_i\}\cup \{s \neq t| s\in T_j, t\in T'_j\}$ for some finite collection of regular sets of terms $S_i,T_j,T'_j$, and show that this is decidable by Theorem \ref{thm:congclos}.  We will then reduce the general case to this case by guessing a set of atomic propositional states to produce only true formulae and a set to produce only false formulae, giving a satisfiability problem in the form of the special case.  We will show that in general $T$ is satisfiable if and only if there exists such a guess which is feasible in the above sense, and from which we can show as a matter of propositional logic that $T$ holds.

\begin{lemma}\label{lem:conjsat}
The problem of determining whether $\bigwedge T$ is satisfiable, where 
\[T = \bigcup_{i=1}^n \{s=t|s,t\in S_i\} \cup \bigcup_{j=1}^m \{s \neq t| s\in T_j, t\in T'_j\},\]
where $S_1,\ldots,S_n,T_1,T'_1,\ldots,T_m,T'_m$ are given regular sets of terms, is decidable.
\end{lemma}
\begin{proof}
Let $E=\bigcup_{i=1}^n \{s=t|s,t\in S_i\}$.  By Proposition \ref{prop:congsound}, $\bigwedge T$ is unsatisfiable if and only if there exists some $j$ and $s\in T_j, t\in T'_j$ with $s\in \congclos_E(\{t\})$; equivalently if $\congclos_E(T_j)\cap \congclos_E(T'_j)\neq \emptyset$ for some $j$.  By Theorem \ref{thm:congclos} we can compute automata $\Aa_i,\Aa'_i$ such that $L(\Aa_i)=\congclos_E(T_i), L(\Aa'_i)=\congclos_E(T'_i)$, and then $\bigwedge T$ is satisfiable if and only if $L(\Aa_i)\cap L(\Aa'_i) = \emptyset$ for all $i$.
\end{proof}

Let $\Aa$ be a tree automaton recognising our theory $T$, and let $\qap$ be the set of atomic propositional states of $\Aa$.  Without loss of generality, for every $q\in \qap$ we have either $L(q)=\{s=t|s\in S, t\in T\}$ or $L(q)=\{s\neq t|s\in S, t\in T\}$ for non-empty regular sets of terms $S,T$.  Let $c$ be a fresh constant symbol, and write $p_\top$ for the vacuously true formula $c=c$ and $p_\bot$ for the formula $c \neq c$.  

For disjoint subsets $Q_\top,Q_\bot \subseteq \qap$, define the automaton $\Aa_{Q_\top,Q_\bot}$ to be $\Aa$ with each atomic propositional state $q$ replaced by an automaton which accepts $\{p_\top\}$ if $q\in Q_\top$, $\{p_\bot\}$ if $q\in Q_\bot$ and $\{p_\top,p_\bot\}$ if $q\in \qap \setminus (Q_\top\cup Q_\bot)$.  Note that the signature of $\Aa_{Q_\top,Q_\bot}$ has no function symbols and only the constant $c$.  Write $\MM_c$ for the unique one-element model corresponding to this signature.

We will call the pair $(Q_\top,Q_\bot)$ of disjoint sets \emph{good} if
\begin{enumerate}[label=(\roman*)]
    \item $\bigwedge T_{Q_\top,Q_\bot}$ is satisfiable, where
    \[T_{Q_\top,Q_\bot} = \bigcup_{q\in Q_\top} L_\Aa(q) \cup \bigcup_{q\in Q_\bot} \{\neg p|p\in L_\Aa(q)\},\] and 
    \item $\MM_c \models \bigwedge L(\Aa_{Q_\top,Q_\bot})$.
\end{enumerate}

We claim that $\bigwedge T$ is satisfiable if and only if there exists a good pair $(Q_\top,Q_\bot)$.  This suffices to prove Theorem \ref{thm:quantfreecon}, since there are only finitely many candidate pairs and for a given pair condition (i) is decidable by Lemma \ref{lem:conjsat} (note that $\bigwedge \{s=t|s\in S, t\in T\}$ is equivalent to $\bigwedge\{s=t|s,t\in S\cup T\}$ if $S$ and $T$ are non-empty) and condition (ii) is decidable by Corollary \ref{cor:modelcheck}.

To prove the claim, suppose that $\bigwedge T$ is satisfiable, so say $\MM \models \bigwedge T$.  Let 
\begin{align*}
Q_\top &= \left\{q\in \qap\mid\MM\models p,\ \forall p\in L_\Aa(q)\right\}\\
Q_\bot &= \left\{q\in \qap\mid\MM\models \neg p,\ \forall p\in L_\Aa(q)\right\}.
\end{align*}
Clearly $\MM\models \bigwedge T_{Q_\top,Q_\bot}$.  For any $\phi \in L(\Aa_{Q_\top,Q_\bot})$, by the construction of $\Aa_{Q_\top,Q_\bot}$ we have $\phi = C[p_1/X_1,\ldots,p_k/X_k]$, where $C$ is purely propositional, the $p_i\in \{p_\top,p_\bot\}$ and in $\Aa$ we have $q_1(X_1),\ldots,q_k(X_k) \rightarrow \qacc(C)$ for some states $q_1,\ldots,q_k\in \qap$ such that for each $q_i$ if $q_i\in Q_\top$ then $p_i=p_\top$ and if $q_i\in Q_\bot$ then $p_i=p_\bot$.  By the definition of $Q_\top$ and $Q_\bot$, if $q\in Q_\top$ there exists $p'_i\in L_\Aa(q)$ with $\MM\models p'_i$, if $q\in Q_\bot$ there exists $p'_i\in L_\Aa(q)$ with $\MM\models \neg p'_i$ and if $q\in \qap \setminus (Q_\top\cup Q_\bot)$ then $L_\Aa(q)$ contains both true and false statements so pick $p'_i\in L_\Aa(q)$ with $\MM\models p'_i$ if $p_i=p_\top$ and $\MM \models \neg p'_i$ if $p_i=p_\bot$.  Then we have $\psi = C[p'_1/X_1,\ldots,p'_k/X_k] \in T$ so $\MM \models \psi$, and for each $i$ we have $\MM\models p'_i \Leftrightarrow \MM_c\models p_i$ so also $\MM_c \models \phi$, as required.

Conversely, suppose that $(Q_\top,Q_\bot)$ is a good pair, and let $\MM \models \bigwedge T_{Q_\top,Q_\bot}$.  For any $\phi\in T$ we have that $\phi = C[p_1/X_1,\ldots,p_k/X_k]$ for some purely propositional $C$ and each $p_i\in L_\Aa(q_i)$ for some $q_i\in \qap$.  Since $\MM\models \bigwedge T_{Q_\top,Q_\bot}$, we have that if $q_i\in Q_\top$ then $\MM\models p_i$ and if $q_i\in Q_\bot$ then $\MM\models \neg p_i$.  Hence by construction of $\Aa_{Q_\top,Q_\bot}$ there exist $p'_1,\ldots,p'_k\in \{p_\top,p_\bot\}$ such that $p'_i = p_\top$ if and only if $\MM\models p_i$ and $\psi = C[p'_1/X_1,\ldots,p'_k/X_k] \in L(\Aa_{Q_\top,Q_\bot})$, so $\MM_c\models \psi$.  Since $\MM_c \models p'_i \Leftrightarrow \MM\models p_i$ we also have $\MM\models \phi$.  Since $\phi\in T$ was arbitrary we have $\MM\models \bigwedge T$ so $\bigwedge T$ is satisfiable, completing the proof of the claim and hence of Theorem \ref{thm:quantfreecon}. \qed



\section{A decidable semantic class: coherent formulae}
\seclabel{coherent}


In this section, we will define a semantic class of purely existential formulae 
(with functions and equality) for which the disjunctive satisfiability problem is decidable.  
This class is inspired by, and as we will show 
in Theorem \ref{thm:cosat}, is a generalisation of, the notion of \emph{coherence} 
for program executions introduced in \cite{Mathur2019}.

\myparagraph{A brief recap of coherent uninterpreted programs}{
The syntax of uninterpreted programs over an FO vocabulary 
$\vocab = (\consts, \funcs, \rels)$ and program variables $\mathcal{X}$ 
is given by the grammar in \figref{grammar}.
An uninterpreted program $P$ does not a priori
fix an interpretation for the symbols in $\Pi$.
The semantics of such a program is instead
determined by fixing an FO model $\model = (U, \interp{})$
that in turn fixes the domain of the variables $\mathcal{X}$ 
as well as the interpretation of the symbols in $\vocab$,
along with an initial assignment to the program variables $\mathcal{X}$.
The verification problem for such a program $P$ against a FO post-condition $\varphi$
(whose free variables belong to $\mathcal{X}$)
asks if there is a model $\model$ and an initial assignment to variables
so that the resulting state after executing the program on $\model$ satisfies $\varphi$~\cite{Mathur2019}.

While undecidable in general, the work in~\cite{Mathur2019} identified
the subclass of \emph{coherent uninterpreted programs} for which the verification problem
becomes decidable, using an automata-theoretic decision procedure.
At a high level, they show that there is a constant space algorithm,
that checks, in a streaming fashion, each sequence of 
statements coming out of the control structure of the program,
and tracks the equality relationships between program elements, but 
not how they relate to previously-computed data.  
Under the conditions of `coherence' of executions, 
the streaming algorithm is both sound and complete --- when 
an execution $\sigma$ is coherent, it is feasible iff the streaming
algorithm accepts it. 
Finally, a constant space algorithm can naturally be compiled to a deterministic
finite state word automaton, which is the key component in the decision procedure
for the verification problem for coherent programs~\cite{Mathur2019}.

Executions of uninterpreted programs can be viewed as if they
were computing terms (over $\vocab$) and storing
them in program variables $\mathcal{X}$.
A \emph{coherent execution}, in particular,
satisfies two 
properties:
\emph{memoization}, which essentially enforces that terms 
cannot be deleted and then recomputed during the course of the execution, and 
\emph{early assumes}, which enforces that whenever a term $t$ is deleted from memory, we cannot later 
see an equality constraint involving a subterm of $t$.
Together, these properties ensure that 
\emph{congruence closure} of the equality constraints can be accurately
performed at each step of the execution, 
by simply tracking equality and disequality relationships between program variables at each step.
In turn, performing accurate congruence closure
is sufficient to check for feasibility of program paths.
}

\myparagraph{From coherent executions to coherent formulae}{
We will see in \secref{exform} that every execution can be translated 
into an existential formula, but of course not every formula arises in this way: in particular, whereas an execution has a linear
structure and can be sequentially processed by a streaming algorithm, this is not the case for a general formula.
Towards this, we first generalize the essence of the streaming algorithm of \cite{Mathur2019} to a a characterization \emph{local consistency}.
at each node of the formula (viewed as a tree) we record the equality relationships between the variables at that node, and check that these are consistent between adjacent nodes.

Local consistency is a regular property, and so (given the undecidability result of~\thmref{exundec}) it 
does not fully capture satisfiability.  
The missing ingredient is a global consistency property, and ts absence can allow
for applying $f$ to equal terms at opposite ends of the formula to yield different answers.  However, it turns 
out that a property similar to memoization is sufficient to ensure that this does not happen: specifically, 
we will enforce that there is no path in the formula along which a particular value is forgotten and then 
recomputed---we will call such a path a \emph{forgetful path}.  With coherence defined in this way it turns 
out that equality constraints propagate backwards such that `early assumes' does not need to be imposed 
separately, and so our decidable class is a strict extension of that in \cite{Mathur2019} even in the 
setting of program executions.
}

\subsection{Coherent formulae}

We first note that throughout this section we make a couple of technical assumptions on the form
of our formulae.  Firstly, we assume that all terms that appear in formulae have depth at most $1$: that is 
they are either variables or have the form $f(x_1,\ldots,x_k)$ for some function $f$.  We denote the set of 
such terms $\shallowterms$ (here $x \in \vars$, $f \in \funcs$):
\begin{align*}
\begin{array}{rcl}
\shallowterms &::=&  x \; | \; f(x, \ldots, x)
\end{array}
\end{align*}
In Appendix~\ref{sec:appshallow},
we show that
any existential regular theory can 
be converted to an existential regular theory of this form.  Secondly, we assume that each formula of our theory 
is purely conjunctive: that is, does not use the symbol $\vee$.  Handling disjunctions is simple 
but adds a slight extra level of complexity and so we relegate it to 
Appendix \ref{sec:disapp}.

We will first define our equivalent of the streaming algorithm from \cite{Mathur2019}: what it 
means for a formula to be locally consistent.
The definition essentially asks if one assign equality relationships between
the different terms and sub-terms of the formula, while
ensuring no contradictions between adjacent nodes.
We first formalize the notion of equality relationships between terms using congruence.
A relation $\sim \subseteq \shallowterms \times \shallowterms$
is said to be a congruence if $\sim$ is an equivalence relation,
and further, whenever $x_1 \sim x'_1, x_2 \sim x'_2 \ldots x_k \sim x'_k$,
then we also have $f(x_1, \ldots, x_k) \sim f(x'_1, \ldots, x'_k)$
(where $f$ has arity $k$).
We use $\mathsf{Congruences}_{\shallowterms}$ to denote the set of all such congruences.
Given a formula $\phi$, a \emph{congruence mapping} 
for $\phi$ (represented as a tree $T$) to be a map 
\[\congr:T \to \mathsf{Congruences}_{\shallowterms},\]
$\congr$ thus records the equality and inequality relationships between 
(a) variables and (b) depth-one terms, for each node of the tree.

\begin{definition}[Local consistency]
Let $\phi$ be a formula and let $\congr$ be a congruence mapping for $\phi$.
$\congr$ is said to \emph{locally consistent} if 
\begin{itemize}
\item for every node $u \in T$ of the form $u = \mlq t_1 = t_2\mrq$, we have
$(t_1, t_2) \in \congr(u)$
\item for every node $u \in T$ of the form $u = \mlq t_1 \neq t_2\mrq$, we have
$(t_1, t_2) \not\in \congr(u)$
\item for every node $u \in T$ of the form $u = \mlq\land(u_1, u_2)\mrq$, we have
$\congr(u_1) = \congr(u_2) = \congr(u)$
\item for every node $u \in T$ of the form $u = \mlq\exists x (u_1)\mrq$, we have
$\congr(u) \wedge \congr(u_1)[x'/x]$ is satisfiable, interpreting $\congr(u)$ and $\congr(u_1)$ as 
propositional formulae over $\shallowterms$
($x'$ is a fresh variable).
\end{itemize}
A formula $\phi$ is locally consistent if there is a locally consistent
congruence mapping $\congr$ for $\phi$.
\end{definition}

\noindent
Local consistency is a necessary but not sufficient condition for satisfiability as we show next.

\noindent
\begin{minipage}{.55\textwidth}
 \begin{example}
    \label{ex:local}
Consider the formula $\phi_= \wedge \phi_\neq$, where
\begin{align*}
    \phi_= &\equiv \exists x'.\left( x' = f(x) \wedge \exists x. \left(x = f(x') \wedge f(x)=x\right)\right) \\
    \phi_\neq &\equiv  \exists x'.\left( x' = f(x) \wedge \exists x. \left(x = f(x') \wedge f(x) \neq x\right)\right)
\end{align*}
The tree corresponding to which is as shown on the right.
Observe that this formula is indeed locally consistent.
However, also observe that it is also plainly unsatisfiable --- writing $x_0$ for the initial 
value of $x$, we have that $\phi_=$ enforces $f^2(x_0)=f^3(x_0)$ 
but $\phi_\neq$ enforces $f^2(x_0)\neq f^3(x_0)$.
\end{example}
\end{minipage}
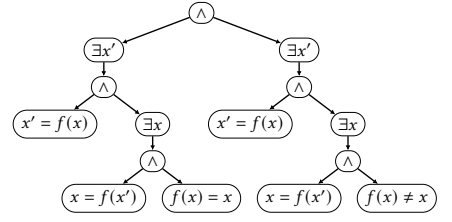
\hfill
\begin{minipage}{.4\textwidth}
\vspace{-0.3in}

\begin{figure}[H]
\newcommand{\ystep}{0.75}
\scalebox{0.65}{
\begin{tikzpicture}
\node (phi) at (0,0) [draw, rounded rectangle] {$\land$};
\node (phi0) at (-2,-1*\ystep) [draw, rounded rectangle] {$\exists x'$};
\node (phi00) at (-2,-2*\ystep) [draw, rounded rectangle] {$\land$};
\node (phi000) at (-3,-3*\ystep) [draw, rounded rectangle] {$x' = f(x)$};
\node (phi001) at (-1,-3*\ystep) [draw, rounded rectangle] {$\exists x$};
\node (phi0010) at (-1,-4*\ystep) [draw, rounded rectangle] {$\land$};
\node (phi00100) at (-2,-5*\ystep) [draw, rounded rectangle] {$x = f(x')$};
\node (phi00101) at (0,-5*\ystep) [draw, rounded rectangle] {$f(x) = x$};
\node (phi1) at (2,-1*\ystep) [draw, rounded rectangle] {$\exists x'$};
\node (phi10) at (2,-2*\ystep) [draw, rounded rectangle] {$\land$};
\node (phi100) at (1,-3*\ystep) [draw, rounded rectangle] {$x'=f(x)$};
\node (phi101) at (3,-3*\ystep) [draw, rounded rectangle] {$\exists x$};
\node (phi1010) at (3,-4*\ystep) [draw, rounded rectangle] {$\land$};
\node (phi10100) at (2,-5*\ystep) [draw, rounded rectangle] {$x=f(x')$};
\node (phi10101) at (4,-5*\ystep) [draw, rounded rectangle] {$f(x)\neq x$};

\draw (phi) edge[-{Latex[length=1mm, width=1mm]}, thick] (phi0);
\draw (phi0) edge[-{Latex[length=1mm, width=1mm]}, thick] (phi00);
\draw (phi00) edge[-{Latex[length=1mm, width=1mm]}, thick] (phi000);
\draw (phi00) edge[-{Latex[length=1mm, width=1mm]}, thick] (phi001);
\draw (phi001) edge[-{Latex[length=1mm, width=1mm]}, thick] (phi0010);
\draw (phi0010) edge[-{Latex[length=1mm, width=1mm]}, thick] (phi00100);
\draw (phi0010) edge[-{Latex[length=1mm, width=1mm]}, thick] (phi00101);
\draw (phi) edge[-{Latex[length=1mm, width=1mm]}, thick] (phi1);
\draw (phi1) edge[-{Latex[length=1mm, width=1mm]}, thick] (phi10);
\draw (phi10) edge[-{Latex[length=1mm, width=1mm]}, thick] (phi100);
\draw (phi10) edge[-{Latex[length=1mm, width=1mm]}, thick] (phi101);
\draw (phi101) edge[-{Latex[length=1mm, width=1mm]}, thick] (phi1010);
\draw (phi1010) edge[-{Latex[length=1mm, width=1mm]}, thick] (phi10100);
\draw (phi1010) edge[-{Latex[length=1mm, width=1mm]}, thick] (phi10101);
\end{tikzpicture}
}
\end{figure}

\end{minipage} \\

We will define a subclass of formulae, which we call `coherent', for which local consistency will 
imply satisfiability.  As discussed above, the condition will be that there is no `forgetful path' in the tree 
along which a particular value is erased and then recomputed: more precisely, the value is stored in a variable at 
both ends of the path, but not at some intermediate point.  
We track this putative value through the 
tree using a family of unary predicates $P_t$, where intuitively $P_t$ expresses that this critical value 
is equal to a shallow term $t \in \shallowterms$ at a particular node.

\begin{definition}[Forgetful path]
Given a pair $(\phi,\congr)$, a \emph{forgetful path} is a pair of nodes $u,v$ of $\phi$ such that there exists
a family of unary predicates $\setpred{P_t}{t\in \shallowterms}$ over
the nodes of $\phi$ with the following properties:
\begin{itemize}
    \item $P=\bigcup_t P_t$ contains the path from $u$ to $v$.
    \item For some variables $x,y$ we have $P_x(u)$ and $P_y(v)$.
    \item For some node $w$ on the path from $u$ to $v$ we have that $\neg P_z(w)$ for all variables z
    \item $\set{P_t}_{t}$ is a minimal healthy family.
\end{itemize}
In the above, we say a family $\set{P_t}_t$ is \emph{healthy} if:
\begin{itemize}
    \item It is non-empty: there exists a term $t$ and node $w$ such that $P_t(w)$ holds.
    \item For each $t,t'$, $\forall w\in \nodes(\phi)$, if $(t,t')\in \congr(w)$ then $P_t(w) \Leftrightarrow P_{t'}(w)$.
    \item For each $t,t'$, $\forall w\in \nodes(\phi)$, if $P_t(w)\wedge P_{t'}(w)$ then $(t,t')\in \congr(w)$.
    \item For every node $w$ of the form $w=\land(w_1, w_2)$ we have $P_t(w_1) \Leftrightarrow P_t(w) \Leftrightarrow P_t(w_2)$ 
    for every term $t$.
    \item For every node $w$ of the form $w = \exists x (w_1)$, we have that 
    \[\congr(w) \wedge \congr(w_1)[x'/x]
    \wedge \bigwedge_{t,t':P_t(w) \wedge P_{t'}(w_1)} t=t'[x'/x]\] is satisfiable.    
\end{itemize}
\end{definition}

Note that the existence of a forgetful path is an MSO property.

\begin{definition}[Coherent Formulae]
\deflabel{coherent-formulae}
A formula $\phi$ is \emph{coherent} if it is either locally inconsistent or has a locally consistent $\congr$ with no forgetful path.  A theory $T$ is coherent if and 
only if $\phi$ is coherent for every $\phi\in T$.
\end{definition}

Since coherence is an MSO property, coherence of a regular theory is decidable in polynomial time 
(for a fixed signature).  

The main theorem of this section 
is that for coherent formulae, local consistency is sufficient for satisfiability.

\begin{theorem}\label{thm:consat}
A coherent formula $\phi$ is satisfiable if and only if it is locally consistent.
\end{theorem}

Since local consistency is a regular property, we immediately have that satisfiability is decidable 
for coherent theories:

\coherentDec*

To prove Theorem \ref{thm:consat} we will consider a pair of equivalence relations on $\nodes(\phi)\times \st$: first the relation 
$\sim$ induced by transitive application of the equalities of $\congr$ on pairs of adjacent vertices; this is tractable 
but insufficient for the existence of a model since it may not respect function application.  
Secondly we will consider the relation $\sim'$ which is $\sim$ enriched by requiring respect for 
function application; this is intractable in general but we will show that if $\phi$ has no forgetful paths 
then in fact $\sim$ and $\sim'$ agree.

Concretely, define $\sim$ to be the least equivalence relation on $\nodes(\phi)\times \st$ consistent with:
\begin{enumerate}[label=(\arabic*)]
    \item\label{prop:congt} If $(t_1,t_2)\in \congr(w)$ then $(w,t_1)\sim (w,t_2)$
    \item\label{prop:andt} If $w$ is of the form $w=\land(w_1, w_2)$ then $(w_1,t)\sim (w,t) \sim (w_2,t)$ for all terms $t$
    \item\label{prop:ext} If $w$ is of the form $w = \exists x (w_1)$ then $(w,t)\sim (w_1,t')$ if $\congr(w) \wedge \congr(w')[x'/x] 
    \models t=t'[x'/x]$.
\end{enumerate}

Define $\sim'$ to be the least equivalence relation containing $\sim$ and consistent with
\begin{enumerate}[label=(4)]
    \item\label{prop:funapp} For each function $f$ of arity $k$, if $(u,x_1)\sim' (v,y_1),\ldots,(u,x_k)\sim' (v,y_k)$ 
    for some tuples of variables $(x_1,\ldots,x_k)$ and $(y_1,\ldots,y_k)$ then $(u,f(x_1,\ldots,x_k)\sim' (v,f(y_1,\ldots,y_k))$.
\end{enumerate}


Note that this containment can be strict. Consider the formula $\phi_=\wedge\phi_\neq$ from Example~\ref{ex:local}. Writing 
$u_=$ and $u_\neq$ for the deepest nodes of $\phi_=$ and $\phi_\neq$ respectively, for any locally consistent 
$\congr$ we have $(u_=,x) \sim' (u_\neq,x)$ but $(u_=,x) \not\sim (u_\neq,x)$.

\begin{proposition}\label{prop:simprime}
    If for every node $w$ of $\phi$ of the form $t_1 \neq t_2$ we have that $(w,t_1)\not\sim' (w,t_2)$ then
    $\phi$ is satisfiable.
\end{proposition}
\begin{proof}
Take the universe to be the set of equivalence classes of $\sim'$, and the interpretation of variable 
$x$ at node $w$ to be $[(w,x)]$.  To define function interpretations, for each function $f$ of arity $k$ and
each tuple of atoms $(a_1,\ldots,a_k)$, if there exists a node $w$ and a tuple of variables $(x_1,\ldots,x_k)$ 
such that $a_i = [(w,x_i)]$ then define $f(a_1,\ldots,a_k)=[(w,f(x_1,\ldots,x_k))]$; this is well-defined 
independent of the choice of $w$ by property \ref{prop:funapp}.  Otherwise, choose $f(a_1,\ldots,a_k)$ arbitrarily. 
Then in this interpretation all equality statements hold by the definition of $\sim'$, and all disequality 
statements hold by the assumption that $(w,t_1) \not\sim' (w,t_2)$ for each disequality statement 
$t_1\neq t_2$.
\end{proof}

\begin{proposition}\label{prop:tildeonevert}
    For any node $w$ and terms $t,t'$ we have $(w,t)\sim (w,t')$ if and only if $(t,t')\in \congr(w)$.
\end{proposition}
\begin{proof}
    Clearly if $(t,t')\in \congr(w)$ then $(w,t)\sim (w,t')$.  For the converse, supposing the contrary there is a 
    sequence of nodes $u_0=w,u_1,\ldots,u_n=w$ and terms $t_0=t,t_1,\ldots,t_n=t'$ such that $(u_i,t_i)\sim 
    (u_{i+1},t_{i+1})$ by a single application of one of rules \ref{prop:andt} or \ref{prop:ext}, possibly 
    followed by an application of rule \ref{prop:congt}, and consider a 
    counterexample with $n$ minimal. Note that $u_1=u_{n-1}$ and by minimality $(t_1,t_{n-1}) \in \congr(u_1)$
    (otherwise $(u_1,t_1)\sim (u_1,t_{n-1})$ is a smaller counterexample);
    then whichever of the final two local consistency conditions applies to the form of the parent node, 
    together with the fact that the corresponding rule \ref{prop:andt} or \ref{prop:ext} could be applied, 
    gives that $(t,t')\in \congr(w)$.
\end{proof}

\begin{proposition}\label{prop:healthtilde}
    If $\{P_t\}$ is a minimal healthy family and $u$ is some node and $t_0$ some term such that $P_{t_0}(u)$ then 
    we have that $P_{t}(v)\Leftrightarrow (u,t_0)\sim(v,t)$. 
\end{proposition}
\begin{proof}
    Observe that $\sim$ induces a healthy family $\{P'_t\}$ by $P'_{t}(v)\Leftrightarrow (u,t_0)\sim(v,t)$ and 
    so the $\Rightarrow$ direction follows by minimality of $\{P_t\}$.  On the other hand, if $(u,t_0)\sim(v,t)$ 
    then there exists a sequence of nodes $u=u_0,\ldots,u_n=v$ and terms $t_1,\ldots,t_n=t$ such that each  
    $(u_i,t_i)\sim (u_{i+1},t_{i+1})$ follows by a single application of one of rules 
    \ref{prop:andt} or \ref{prop:ext}, perhaps followed by an application of rule 
    \ref{prop:congt}.  But from each of these we can conclude inductively that $P_{t_i}(u_i)$.
\end{proof}

Note that it immediately follows that any $\{P_t\}$ defined as $P_t(u) \Leftrightarrow (u,t)\sim (u_0,t_0)$ for fixed 
$u_0$ and $t_0$ is minimal, since otherwise consider a minimal healthy family $\{P'_t\}$ contained in $\{P_t\}$.  Then by Proposition 
\ref{prop:healthtilde} we have $\{P'_t\}=\{P_t\}$.

We are now ready to prove the main lemma:

\begin{lemma}\label{lem:badpaths}
    If $(\phi,\congr)$ has no forgetful paths then $\sim$ and $\sim'$ agree.
\end{lemma}
\begin{proof}
If $\sim$ and $\sim'$ do not agree then there is an opportunity to apply rule \ref{prop:funapp}, i.e. some nodes $u,v$ and 
variables $(x_1,\ldots,x_k)$ and $(y_1,\ldots,y_k)$ such that $(u,x_i) \sim (v,y_i)$ but $(u,f(x_1,\ldots,x_k)) \not\sim (v,f(y_1,\ldots,y_k))$.
Let $X_i$ be the set of nodes $w$ such that $(u,x_i) \sim (w,z)$ for some variable $z$ and $Y_i \supseteq X_i$ the set of nodes $w$
such that $(u,x_i) \sim (w,t)$ for some term $t$.  Then for each $i$, by considering the sequence of nodes $(u_0,\ldots,u_n)$ as in the 
proof of Proposition \ref{prop:healthtilde} we have that $Y_i$ contains the path from $u$ to $v$.
%
On the other hand, we claim that if $X_i$ contains the path from $u$ to $v$ for all $i$ then $(u,f(x_1,\ldots,x_k)) 
\sim (v,f(y_1,\ldots,y_k))$.  Indeed, let the path be $u=u_0,\ldots,u_n=v$.  Write $z_{j,1},\ldots,z_{j,k}$ for variables 
such that $(u,x_i)\sim (u_j,z_{j,i})$ and suppose for induction that $(u,f(x_1,\ldots,x_k)) \sim (u_j,f(z_{j,1},\ldots,z_{j,k}))$.
Then applying rule \ref{prop:andt} or \ref{prop:ext} as appropriate to the pair of nodes $u_j,u_{j+1}$ gives that
$(u,f(x_1,\ldots,x_k)) \sim (u_{j+1},f(z_{j+1,1},\ldots,z_{j+1,k}))$.  Hence by induction $(u,f(x_1,\ldots,x_k)) 
\sim (u_n=v,f(z_{n,1},\ldots,x_{n,k})) \sim (v,f(y_1,\ldots,y_k))$, where the last $\sim$ is by the fact that 
$(v,z_{n,i}) \sim (u,x_i) \sim (v,y_i)$ and hence by Proposition \ref{prop:tildeonevert} we have $(z_{n,i},y_i)\in \congr(v)$
so $(f(z_{n,1},\ldots,z_{n,k}),f(y_1,\ldots,y_k)) \in \congr(v)$.  So the claim is proved and we have some $i$ such that 
$X_i$ does not contain the path from $u$ to $v$.

Thus, for some $i$, there is a node $w\in Y_i\setminus X_i$ on the path from $u$ to $v$; let $t_0$ be a term such that 
$(u,x_i)\sim (w,t_0)$.  Define the family $\{P_t\}$ by $P_t(w') \Leftrightarrow (w',t)\sim (w,t_0)$.  This is a 
healthy family and is minimal by the remark following Proposition \ref{prop:healthtilde}. 
Since $Y_i$ contains the path from $u$ to $v$ 
we have that $P=\bigcup_t P_t$ contains the path; we also have that $P_{x_i}(u)$ and $P_{y_i}(v)$ but $\neg P_z(w)$ 
for all variables $z$, so $u,v$ is a forgetful path, as required.
\end{proof}

Theorem \ref{thm:consat} follows immediately from Lemma \ref{lem:badpaths} and Proposition \ref{prop:simprime}.

\subsection{Executions as formulae}\label{sec:exform}

We now show that definition of coherence given above generalises the notion of coherence for uninterpreted programs 
from \cite{Mathur2019}: we will show how to translate executions of uninterpreted programs into FO formulae, and 
that for any execution which is coherent in the sense of \cite{Mathur2019}, the resulting FO formula is coherent 
in the sense of this section.

An execution in the sense of \cite{Mathur2019} is a finite sequence of statements of the following three types:
``$\assume(x=y)$'', ``$\assume(x\neq y)$'', ``$x:=t$'',  where $x$ and $y$ are variables and $t\in \shallowterms$.
We translate to a formula over a signature which contains for each variable $x$ an additional dummy variable $x'$, and 
define the translation function $\psi$ inductively as follows:
\begin{align*}
    \psi(\assume(x=y)\cdot \sigma) &= x=y\wedge \psi(\sigma) \\
    \psi(\assume(x\neq y)\cdot \sigma) &= x \neq y \wedge \psi(\sigma) \\
    \psi(x := t \cdot \sigma) &= \exists x'.(x' = t \wedge \exists x. (x = x' \wedge \psi(\sigma))) \\
    \psi(\epsilon) &= \top
\end{align*}


\begin{theorem}\label{thm:cosat}
    Let $\sigma$ be a coherent execution.  Then $\psi(\sigma)$ is a coherent formula.
\end{theorem}

\begin{proof}
We will show that the minimal $\congr$ has no forgetful paths; in this proof we assume that the reader is familiar with 
concepts from~\cite{Mathur2019}.

A preliminary observation is that it suffices to consider only `joint nodes' which are the final node arising from a particular instruction (that is, the 
node where the recursive call $\psi(\sigma)$ occurs in the translation above).  \dm{Is the meaning of this clear?} \mv{maybe. we can leave it for now.}
Indeed, suppose we have a bad path beginning or ending at a node $u$ which is not a joint node.  If $u$ is inside the translation
of an $\assume$ statement then moving the endpoint of the path to the next joint node above or below $u$ will still give 
a bad path, since the translation of an $\assume$ statement does not contain any quantifiers.  If $u$ is inside the translation 
of an assignment, if $u$ is not between the two existential quantifiers then we can move it up or down as before.  Otherwise, 
we have $P_z(u)$ for some variable $z$, and if $z \neq x'$ then we can move the endpoint up to the next joint node, while if 
$z \neq x$ we can move it down to the next joint node.

For any joint node $w$, write $\sigma_w$ for the execution up to $w$.  Define $\sim_w$ to be the relation $\sim$ as in the previous section 
but considering only the tree truncated at $w$.  Then whenever $w$ is before $v$ we have that $\sim_v$ 
contains $\sim_w$, and if $w$ is the final node then $\sim_w = \sim$.


Suppose that the formula truncated at $w$ is coherent.  Then, $\sim_w$ and $\sim'_w$ agree, and 
$(u,x) \sim_w (v,y)$ if and only if $\comp(\sigma_u,x) \cong_{\alpha(\sigma_v)} \comp(\sigma_v,y)$ (\cite[Proposition 1]{Mathur2019}).

Now suppose that $u,v$ is a forgetful path with $v$ as early as possible.  
Say $\sigma_v = \sigma_{v'} \cdot y := t_v$; then we have that 
$y$ is the variable for which $P_y(v)$, and also that in the final node $v'$ of $\sigma'_v$ we have 
$\neg P_z(v')$ for all variables $z$, and that the formula truncated at $v'$ is coherent.  Now we have 
$\comp(\sigma_{v'},z) \not \cong_{\alpha(\sigma_{v'})} \comp(\sigma_v,y)$ for all $z$ (by Proposition \ref{prop:healthtilde}), so if 
$\comp(\sigma_u,x) \cong_{\alpha(\sigma_v)} \comp(\sigma_v,y)$ we have a memoization violation.  
For every variable $z$ 
other than $y$ we have $\comp(\sigma_u,x) \not \cong_{\alpha(\sigma_v)} \comp(\sigma_v,z)$ since otherwise 
$\comp(\sigma_{v'},z) = \comp(\sigma_v,z) \cong_{\alpha(\sigma_{v'})} \comp(\sigma_u,x)$ so $(u,x) \sim'_{v'} (v',z)$, contradiction.

Suppose on the other hand $\comp(\sigma_u,x) \not \cong_{\alpha(\sigma_v)} \comp(\sigma_v,z)$ for all variables $z$, but for some later node $w$ we have
$\comp(\sigma_u,x) \cong_{\alpha(\sigma_w)} \comp(\sigma_v,y)$.  Let $w$ be the first such node, and write $\sigma_w = \sigma_{w'}\cdot \assume(z=z')$.
Then $\comp(\sigma_u,x)$ must be a superterm of $\comp(\sigma_{w'},z)$ modulo $\cong_{\alpha(\sigma_{w'})}$ (WLOG $z$ rather than $z'$), 
so by the early assumes property there is a variable $z''$ such that $\comp(\sigma_{w'},z'') \cong_{\alpha(\sigma_{w'})} \comp(\sigma_u,x)$.
But this gives rise to a memoization violation since $\comp(\sigma_u,x)$ is stored in a variable after $\sigma_x$, then 
is not after $\sigma_v$, and then is again after $\sigma_w$. \sloppy
\end{proof}

Note that the converse of Theorem \ref{thm:cosat} does not hold: there exist executions $\sigma$ which 
are not coherent in the sense of \cite{Mathur2019}, but for which $\psi(\sigma)$ is coherent.  For 
example, one of the examples in \cite{Mathur2019} of a non-coherent execution is 
$\sigma=z := f(x) \cdot z := f(z) \cdot \assume(x = y)$ (which fails early assumes), but it is 
easy to see that $\psi(\sigma)$ is coherent.  This means that even in the setting of uninterpreted 
programs our definition of coherence generalises that in \cite{Mathur2019}.


\section{Related work}
\seclabel{related}

The classical decision problem of determining the satisfiability/validity of 
a given first-order sentence is undecidable~\cite{god30,chu36,tur37}. 
Investigations into decidable fragments led to an almost complete classification 
of decidable fragments, summarized in~\cite{classical-book}. 
In this paper, we look at theories belonging to two of these classes, 
for which the classical decision problem is decidable:
the EPR (or Bernays-Sch{\"o}nfinkel) class consisting of 
function-free formulae with quantifiers of the form $\forall^*\exists^*$ 
and shown decidable by Ramsey~\cite{ramsey1930problem}; 
and the Gurevich class consisting of purely existential formulae 
and shown decidable by Gurevich~\cite{gurevich1976decision}. 
In both cases, decidability is established through the finite model property. 
We show that the problems of
conjunctive and disjunctive satisfiability of theories of formulae from 
these classes are both undecidable in general. 
We then identify sub-fragments of these classes for which these problems become decidable.
Blumensath and Gr{\"a}del~\cite{blumensath2004finite} consider the problem
of model checking an infinite structure (presented finitely) 
against a single formula in first order logic, and is orthogonal to our work.

Our regular theories consist of formulae over finitely many variables. Properties of finite variable logics have also been extensively studied. It is known that the classical decision problem is undecidable for formulae when the number of variables is 3 or more~\cite{kmw62}, but the problem is {\nexp}-complete for two-variable logics~\cite{mor75,gkv97}. Decidability is once again established due to the finite model property. Limiting the number of variables, limits the ability to count within the logic. Therefore, finite variable logics have been extended with counting quantifiers that allow one to say that ``there are at least $k$ elements''~\cite{imm}. 
Two-variable logics with counting quantifiers also admit decidable satisfiability~\cite{gor97}, 
even though they do not enjoy the finite model property. 
Extensions of two-variable logics with transitive closure predicates leads to undecidability~\cite{gor99}.

Satisfiability modulo theories (SMT) considers the problem of determining the satisfiability of formulas where functions and relations need to satisfy special properties often identified through first order theories. A number of decidable theories have been identified and these include linear arithmetic over rationals and integers, fixed-width bit vectors~\cite{bit-vectors}, floating point arithmetic~\cite{floating-point1,floating-point2}, strings~\cite{strings}, arrays~\cite{arrays}, and uninterpreted functions~\cite{shostak1978algorithm}.
The work of Krogmeier and Madhusudan~\cite{Krogmeier2022,Krogmeier2023} considers the problem
of learning logical formulae over an infinite space of FO and SMT formulae from a finite
set of models, and find applications in unrealizability of synthesis problems in various contexts~\cite{coherent-synthesis2020,Hu-unrealizability-2019,hu-unrealizability-2020} or in synthesis of axiomatizations~\cite{krogmeier2022axiomatizations}.
Combinations of SMT theories have also been
shown to be decidable under certain conditions~\cite{EQSMT2018,Tinelli1996,NelsonOppen1979}.


\section{Conclusions and Future Work}

Our work is a first attempt at a natural generalisation of the classical decision problem, 
also known as the `Entscheidungsproblem', that asks if a given first order logic formula is valid, 
to the case of an infinite set of formulae, presented effectively using a finite state automaton.

We have considered two of the six maximal decidable classes for the classical decision  problem, the EPR class and the Gurevich class.  We showed that each of these classes is undecidable 
for regular theories, but that each contains a decidable subclass.  Nevertheless, this is only a 
glimpse of the landscape within each of these classes, and much remains to explore: for example, what 
about the classes $[\exists^*\forall,\allc,(0)]_=$ and $[\exists^*\forall^2,\allc,(0)]_=$ inside the EPR class? 

A more challenging prospect is to investigate the remaining four decidable classes.  For example, 
Rabin's class $[\allc,(\omega),(1)]_=$ tells us that for the classical decision problem everything is 
decidable if we are restricted to unary predicates and a single unary function.  Is this also the 
case for regular theories?  At the very least the proof will have to be rather different, since 
the proof for single formulae goes by Rabin's theorem that MSO on trees is decidable, which as 
we have seen in Theorem \ref{thm:eprundec} fails for regular theories. 
Another interesting direction is to investigate decidability for the case of second order logic, possibly borrowing upon decidable classes such as {\tt EQSMT}~\cite{EQSMT2018} or first order logic with least fixpoints.

The formulation of the classical decision problem in terms of regular theories serves as yet another connection between logic and automata theory, and draws upon the rich literature and seminal results in these two fields, while also opening up new and interesting directions of research in both areas. 
We also believe that viewing problems in areas such as program verification and synthesis from the lens of this generalization is both natural and can pave the way for interesting research directions.
A popular paradigm in algorithmic (i.e., completely automatic) verification is to view programs as sets of program paths is often regular (also see~\secref{motivating}). A popular class of verification techniques, namely  trace abstraction (and its refinement)~\cite{traceAbstraction2009} crucially uses this insight and has been effective in pushing the boundary of practical automated verification~\cite{ultimate-automizer-2013,HeizmannAutomata2013} and test generation~\cite{ultimate-test-gen-2024}. We opt for a similar view, but align it more closely with logic-driven approaches such as symbolic execution~\cite{SymbExec1976,KLEE2008} where each program path is viewed as a formula that characterises its feasibility. We think this connection will be interesting to explore and may result into new classes of programs for which program verification and synthesis is decidable.


\begin{acks}
We thank anonymous reviewers for their constructive feedback on
an earlier drafts of this manuscript. 
Umang Mathur was partially supported by the Simons Institute for the Theory of Computing, 
and by a Singapore Ministry of Education (MoE) Academic Research Fund (AcRF) Tier 1 grant.
Mahesh Viswanathan are partially supported by NSF SHF 1901069 and NSF CCF 2007428. 
\end{acks}

\bibliographystyle{ACM-Reference-Format}
\bibliography{refs}

\clearpage
\newpage
\appendix


\section{Regularity of congruence closure}\label{sec:congreg}

\begin{theorem}[Theorem \ref{thm:congclos}]
Let $\Aa,\Aa_1,\ldots,\Aa_n$ be tree automata.  Let $E=\{s=t|s,t\in L(\Aa_i) \text{ for some $i$}\}$.  Then $\congclos_E(L(\Aa))$ is regular, and given by an automaton computable as a function of $\Aa,\Aa_1,\ldots,\Aa_n$.
\end{theorem}

Fix $S=L(\Aa)$ and $S_i=L(\Aa_i)$.  The algorithm to compute $\congclos_E(S)$ proceeds as follows:
\begin{itemize}
    \item Form the automaton $\Aa'=\Aa \sqcup \Aa_1 \sqcup \ldots \sqcup \Aa_n$, with accepting state the accepting state of $\Aa$.  Let the accepting states of the $\Aa_i$ be denoted by $q_i$.
    \item Repeatedly apply the $\grow$ operation to $\Aa'$ until stability, where grow consists of testing for every state $q$ and every $i$ whether $L(q)\cap S_i$ is non-empty, and if so adding a silent transition $q_i \rightarrow q$ (if not already present).
\end{itemize}

This clearly terminates, since $\grow$ adds new transitions but not new states, so it remains to prove that $L(\grow^\omega(\Aa')) = \congclos_E(S)$.  In order to prove that this operation on automata corresponds to the set of terms for which there exist proofs of congruence using the transitivity and congruence rules, we will define two distinct but related functions on sets of terms, show that both have congruence closure as their least fixed point, and finally use their relationship to $\grow$ to show that $\grow^\omega$ also gives the congruence closure.

The first function, $\subst$, is given by substituting subterms of the input term which are in some $S_i$ with another element of $S_i$, possibly for several subterms at the same time (and several $i$); more formally, for any set of terms $L$ define 
\begin{align*}
\subst&(L) = \{t\in \Ss^*|\exists s\in L \text{ s.t. $\exists$ context $C$ and terms } t_1,t'_1,\ldots,
t_k,t'_k \\ &\text{ s.t. for all $i$, }
t_i,t'_i\in S_j \text{ for some $j$, and } s=C[t_1/X_1,\ldots,t_k/X_k]\\ 
&\text{ and } t=C[t'_1/X_1,\ldots,t'_k/X_k]\}.
\end{align*}

\begin{lemma}\label{lem:substom}
For any set $L$, $\subst^\omega(L)=\congclos_E(L)$, where $E=\{s=t|s,t \in S_i \text{ for some $i$}\}$.
\end{lemma}
\begin{proof}
Clearly $\subst^\omega(L) \subseteq \congclos_E(L)$: indeed, $\subst(L) \subseteq \congclos_E(L)$ (since $t_i=t'_i \in E$ for all $i$ and hence $(s,t)\in R$ by congruence), so also $\subst^\omega(L) \subseteq \congclos_E(L)$.

For the converse, note that it suffices to consider singleton sets $L=\{t\}$.  We will proceed by induction on the length of the congruence proof, but some care is required: we have to perform the induction simultaneously for all singletons $t$.  Thus our inductive hypothesis is `for all $t$, $\kcc{k}_E(\{t\}) \subseteq \subst^\omega(\{t\})$', where $\kcc{k}(L)$ is the set of elements with congruence closure proofs from $L$ of length at most $k$.  The base case is trivial since $\kcc{0}(L)=L$.

Now suppose that $t'\in \kcc{(k+1)}(\{t\})$.  The final step of the congruence proof is either
\[\text{(trans): }t=s, s=t \Rightarrow t=t'\]
or
\[\text{(cong): }s_1=s'_1,\ldots,s_m=s'_m \Rightarrow t=f(s_1,\ldots,s_m)=f(s'_1,\ldots,s'_m)=t',\]
and we consider the two cases separately.

Case (trans): since $t=s$ and $s=t'$ appeared as conclusions in the first $k$ steps of the proof, we have $s\in \kcc{k}(\{t\})$ and $t'\in \kcc{k}(\{s\})$.  Hence by the inductive hypothesis $s\in \subst^\omega(\{t\})$, so $\subst^\omega(\{s\})\subseteq \subst^\omega(\{t\})$.  Also by the inductive hypothesis we have $t'\in \subst^\omega(\{s\})$ and so $t'\in \subst^\omega(\{t\})$, as required.

Case (cong): by the inductive hypothesis we have $s'_i\in \subst^\omega(\{s_i\})$, and hence $t'\in f(\subst^\omega(\{s_1\}),\ldots,\subst^\omega(\{s_m\})) \subseteq \subst^\omega(\{f(s_1,\ldots,s_m)\}) = \subst^\omega(\{t\})$, as required. \sloppy
\end{proof}

Now it is easy to see that $\subst(L(\grow^k(\Aa'))) \subseteq L(\grow^{k+1}(\Aa'))$.  Indeed, if $t\in \subst(L(\grow^k(\Aa')))$ then $\exists s\in L(\grow^k(\Aa'))$ and context $C$ and terms $t_1,t'_1,\ldots,t_k,t'_k$ with each $t_i,t'_i \in S_{j_i}$ for some $j_i$ and $s=C[t_1/X_1,\ldots,t_m/X_m]$, $t=C[t'_1/X_1,\ldots,t'_m/X_m]$.  Since $s\in \grow^k(\Aa')$ there must exist states $q'_1,\ldots,q'_m$ such that in $\grow^k(\Aa')$ we have $t_i\in L(q'_i)$ and $q'_1(X_1),\ldots,q'_m(X_m)\rightarrow \qacc(C)$.  Since $t_i \in L(q'_i)\cap S_{j_i}$ we have that in $\grow^{k+1}(\Aa')$ there is a silent transition $q_{j_i} \rightarrow q'_i$ and so $S_{j_i}\subseteq L(q'_i)$ and in particular $t'_i\in L(q'_i)$, and hence $t\in L(\grow^{k+1}(\Aa'))$, as required.  Combining this with Lemma \ref{lem:substom} gives
\begin{equation}\label{eq:substom}\subst^\omega(L(\Aa')) = \congclos_E(L(\Aa')) \subseteq L(\grow^\omega(\Aa')).\end{equation}

For the reverse containment, we define a slightly stronger function $\subst'$, which allows substituting subterms in $S_i$ not only by other elements of $S_i$ but of $\subst^\omega(S_i)$; formally, for any set $L$ define

\begin{align*}
\subst'&(L) = \{t\in \Ss^*|\exists s\in L \text{ s.t. $\exists$ context $C$ and terms } t_1,t'_1,\ldots,
t_k,t'_k\\ &\text{ s.t. for all $i$, }
t_i\in S_j, t'_i \in \subst^\omega(S_j) \text{ for some $j$, and } \\ &s=C[t_1/X_1,\ldots,t_k/X_k] \text{ and }
t=C[t'_1/X_1,\ldots,t'_k/X_k]\}.
\end{align*}

Unsurprisingly, a single application of $\subst'$ can be obtained as a repeated application of $\subst$ and so it has the same least fixed point as $\subst$:

\begin{lemma}\label{lem:substpom}
For any set $L$, $\subst'^\omega(L)=\subst^\omega(L)$.
\end{lemma}
\begin{proof}
Trivially $\subst'^\omega(L) \supseteq \subst^\omega(L)$.  For the converse, it suffices to prove $\subst'(L) \subseteq \subst^\omega(L)$.  We will prove by induction on $k$ that $\subst'_k(L) \subseteq \subst^\omega(L)$, where $\subst'_k$ is defined as $\subst'$ where each of the $t'_i\in \subst^k(S_j)$.  The base case is trivial since $\subst'_0=\subst$.

Now say $t\in \subst'_{k+1}(L)$, so $t=C[t'_1/X_1,\ldots,t'_m/X_m]$ for some $t'_i\in \subst^{k+1}(S_{j_i})$, and $C[t_1/X_1,\ldots,t_m/X_m] \in L$ for some $t_i \in S_{j_i}$.  Since each $t'_i\in \subst^{k+1}(S_{j_i})$, we have $t'_i=C_i[t'_{i,1}/X_{i,1},\ldots,t'_{i,m_i}/X_{i,m_i}]$ and $C_i[t_{i,1}/X_{i,1},\ldots,t_{i,m_i}/X_{i,m_i}]\in \subst^k(S_{j_i})$ for some context $C_i$ and terms $t_{i,j}, t'_{i,j}\in S_{k_{i,j}}$ for some $S_{k_{i,j}}$.  But then we have
\[t = (C[C_1/X_1,\ldots,C_m/X_m])[t'_{1,1}/X_{1,1},\ldots,t'_{2,1}/X_{2,1},\ldots],\]
and 
\begin{align*}
    t' := (C[C_1/&X_1,\ldots,C_m/X_m])[t_{1,1}/X_{1,1},\ldots,t_{2,1}/X_{2,1},\ldots,t_{m,1}/X_{m,1},\ldots] \\ 
    &= C[C_1[t_{1,1}/X_{1,1},\ldots]/X_1,\ldots,C_m[t_{m,1}/X_{m,1},\ldots]/X_m] \\
    &\in \subst'_k(L)
\end{align*}
since each of the $C_i[t_{i,1}/X_{i_1},\ldots]\in \subst^k(S_{j_i})$ and $C[t_1/x_1,\ldots]\in L$.  By the inductive hypothesis $\subst'_k(L)\subseteq \subst^\omega(L)$ and so $t'\in \subst^\omega(L)$ and $t\in \subst(\{t'\}) \subseteq \subst^\omega(L)$, as required.
\end{proof}

We are now ready to complete the proof of Theorem \ref{thm:congclos}.  We will prove by induction on $l$ that for any state $q$ we have $L_{\grow^l(\Aa')}(q) \subseteq \subst^\omega(L_\Aa'(q))$.  To do this, for fixed $l$ we show by induction on $k$ that for every term $t$ with an accepting run of size at most $k$, if $t\in L_{\grow^{l}(\Aa')}(q)$ then $t\in \subst^\omega(L_{\Aa'}(q))$.

Let $t$ be such a term.  Unless $t\in L_{\grow^{l-1}}(\Aa')(q)$, we have $t=C[t'_1/X_1,\ldots,t'_m/X_m]$, where $t'_i\in L_{\grow^l(\Aa')}(q_{j_i})$ for some $j_i$, $q_{j_i}\rightarrow q'_i$ are new silent transitions inserted by the $l$th application of $\grow$ and in $\grow^{l-1}(\Aa')$ we have $q'_1(X_1),\ldots,q'_m(X_m)\rightarrow q(C)$ and this forms part of a minimum-size run for $q(t)$ in $\grow^l(\Aa')$ (intuitively, we are cutting up a minimum accepting run for $t$ at the first places newly-added transitions are taken). 

Now since the $q_{j_i}\rightarrow q'_i$ transitions were inserted, we must have $t_i\in L_{\grow^{l-1}(\Aa')}(q'_i)$ for some $t_i\in S_{j_i}$.  Since $q'_1(X_1),\ldots,q'_m(X_m) \rightarrow q(C)$ in $\grow^{l-1}(\Aa')$, we have that $C[t_1/X_1,\ldots,t_m/X] \in L_{\grow^{l-1}(\Aa')}(q) \subseteq \subst^\omega(L_{\Aa'}(q))$ by the inductive hypothesis on $l$. Also the $t'_i$ must have accepting runs in $L_{\grow^l(\Aa')}(q_{j_i})$ of size less than $k$, and so by the inductive hypothesis on $k$ we have $t'_i\in \subst^\omega(L_{\Aa'}(q_{j_i})) = \subst^\omega(S_{j_i})$.  Hence we have that $t\in \subst'(\subst^\omega(L_{\Aa'}(q))) = \subst^\omega(L_{\Aa'}(q))$ by Lemma \ref{lem:substpom}, as required.\sloppy

Taking $q=\qacc$ gives $L(\grow^\omega(\Aa')) \subseteq \subst^\omega(L(\Aa'))$.  Combining this with equation (\ref{eq:substom}) completes the proof of Theorem \ref{thm:congclos}. \qed

\section{Reduction to shallow terms}\label{sec:appshallow}

Without loss of generality say that all functions symbols in our signature have arity $k$, and 
let $x_1,\ldots,x_k$ be fresh.  Given a grammar with production rules of the form
\[X := f_1(X_{1,1},\ldots X_{1,k}) | \ldots | f_m(X_{m,1},\ldots,X_{m,k})|y_1|\ldots |y_n,\]
define the grammar
\begin{align*}
    \Psi_X(x) := &\big|_{i=1}^m \exists x_1,\ldots,x_k. x=f_i(x_1,\ldots,x_k) \wedge \Psi_{X_{i,1}}(x_1) 
    \wedge \ldots \wedge \Psi_{X_{i,k}}(x_k) \\
    &| x=y_1 | \ldots | x=y_n.
\end{align*}

Then in a given model $\model$ we have that $\model \models x=t$ for some $t\in X$ if and only if $\model \models \psi$ for some $\psi \in \Psi_X(x)$.

Hence in the grammar defining our regular theory, whenever we have $X=X'$, we can replace it with 
\[\exists x_1,x_2.x_1=x_2 \wedge \Psi_X(x_1) \wedge \Psi_{X'}(x_2),\]
and similarly replace $X\neq X'$ with 
\[\exists x_1,x_2.x_1\neq x_2 \wedge \Psi_X(x_1) \wedge \Psi_{X'}(x_2).\]


\section{Formulae with disjunction}\label{sec:disapp}

In Section \ref{sec:coherent}, we assumed that our formulae were purely conjunctive, i.e. did not 
contain the symbol $\wedge$; we now show that the definitions of local consistency and coherence 
can be trivially generalised to this setting.  We enrich $\congr$ with a function $\val:T\rightarrow \{0,1\}$, 
where $\val(u)$ denotes whether the subformula at $u$ is known to hold.  The local consistency conditions 
are now:
\begin{itemize}
\item for every node $u \in T$ of the form $u = `t1 = t2'$, we have
$\val(u) \Rightarrow (t_1, t_2) \in \congr(u)$
\item for every node $u \in T$ of the form $u = `t1 \neq t2'$, we have
$\val(u) \Rightarrow (t_1, t_2) \not\in \congr(u)$
\item for every node $u \in T$ of the form $u = `\land(u1, u2)'$, we have
$\congr(u_1) = \congr(u_2) = \congr(u)$ and $\val(u) \Rightarrow \val(u_1) \wedge \val(u_2)$
\item for every node $u \in T$ of the form $u = `\lor(u1, u2)'$, we have
$\congr(u_1) = \congr(u_2) = \congr(u)$ and $\val(u) \Rightarrow \val(u_1) \vee \val(u_2)$
\item for every node $u \in T$ of the form $u = `\exists x (u_1)'$, we have
$\congr(u) \wedge \congr(u_1)[x'/x]$ is satisfiable, interpreting $\congr(u)$ and $\congr(u_1)$ as 
propositional formulae over $\shallowterms_1$, and $\val(u) \Rightarrow \val(u_1)$
\item for $u$ the root node of $T$, we have $\val(u)=1$
\end{itemize}

Then in the definition of coherence (and correspondingly in the proof of Theorem \ref{thm:consat}) we 
just completely ignore the parts of the tree for which $\val$ is $0$.  Concretely, in the definition of 
a bad path we modify the healthiness conditions on $\{P_t\}$ to say that $\neg P_t(u)$ whenever 
$\val(u)=0$, the other healthiness conditions are only imposed for nodes where $\val(u)=1$, and the 
condition for disjunctive nodes is
\begin{itemize}
    \item For every node $w$ of the form $w = \lor(w_1,w_2)$ we have that if $\val(w_i)=\val(w)=1$ then 
    we have $P_t(w_i) \Leftrightarrow P_t(w)$ for every term $t$.
\end{itemize}

\end{document}
\endinput